\documentclass[11pt]{article}
\usepackage{amssymb,amstext,amsmath}
\usepackage{amsthm}
\usepackage{mathrsfs}
\usepackage{mathtools}
\usepackage{graphicx}
\usepackage{float}
\usepackage[english]{babel}

\usepackage[margin=1in]{geometry}
\newtheorem{theorem}{Theorem}
\newtheorem*{theorem*}{Theorem}

\newtheorem{lemma}{Lemma}

\newtheorem*{conjecture}{Conjecture}

\newtheorem*{HM}{The Hidden Matching problem ($\operatorname{HM}^\alpha_n$)}
\newtheorem*{BHM}{The Boolean Hidden Matching problem ($\operatorname{BHM}^\alpha_n$)}
\newtheorem*{BHH}{The Boolean Hidden Hypermatching problem ($\operatorname{BHH}^\alpha_{t,n}$)}
\newtheorem*{BHP}{The $f$-Boolean Hidden Partition problem ($f\operatorname{-BHP}^{\alpha}_{t,n}$)}

\usepackage[pdftex,colorlinks=true,linkcolor=blue,citecolor=blue,urlcolor=black]{hyperref}

\newcommand{\qedalt}{\tag*{$\blacksquare$}}

\newtheorem{innercustomthm}{Theorem}
\newenvironment{customthm}[1]
  {\renewcommand\theinnercustomthm{#1}\innercustomthm}
  {\endinnercustomthm}
 
\newenvironment{myproof}[1] {\paragraph{Proof of {#1}.}}{}

\usepackage[dvipsnames]{xcolor}
\usepackage{tasks}
\settasks{
	counter-format=(tsk[r]),
	label-width=4ex
}

\usepackage{lipsum,nccmath}
\usepackage{thm-restate}

\DeclareMathOperator{\sgn}{sgn}

\newcommand{\be}{\begin{equation}}
\newcommand{\ee}{\end{equation}}
\newcommand{\bea}{\begin{eqnarray}}
\newcommand{\eea}{\end{eqnarray}}
\newcommand{\bes}{\begin{equation*}}
\newcommand{\ees}{\end{equation*}}
\newcommand{\beas}{\begin{eqnarray*}}
\newcommand{\eeas}{\end{eqnarray*}}
\newcommand{\ba}{\begin{align}}
\newcommand{\ea}{\end{align}}
\newcommand{\bas}{\begin{align*}}
\newcommand{\eas}{\end{align*}}

\begin{document}

\title{Exponential quantum communication reductions from generalisations of the Boolean Hidden Matching problem}
	
\author{Jo\~{a}o F. Doriguello\footnote{School of Mathematics and Quantum Engineering Centre for Doctoral Training, University of Bristol. \url{joao.doriguellodiniz@bristol.ac.uk}} \and Ashley Montanaro\footnote{School of Mathematics, University of Bristol.}}

	\maketitle

	\begin{abstract}
	    In this work we revisit the Boolean Hidden Matching communication problem, which was the first communication problem in the one-way model to demonstrate an exponential classical-quantum communication separation. In this problem, Alice's bits are matched into pairs according to a partition that Bob holds. These pairs are compressed using a Parity function and it is promised that the final bit-string is equal either to another bit-string Bob holds, or its complement. The problem is to decide which case is the correct one. Here we generalise the Boolean Hidden Matching problem by replacing the parity function with an arbitrary Boolean function $f$. Efficient communication protocols are presented depending on the sign-degree of $f$. If its sign-degree is less than or equal to 1, we show an efficient classical protocol. If its sign-degree is less than or equal to $2$, we show an efficient quantum protocol. We then characterize the classical hardness of all symmetric functions $f$ of sign-degree greater than or equal to $2$, except for one family of specific cases. We also prove, via Fourier analysis, a classical lower bound for any function $f$ whose pure high degree is greater than or equal to $2$. Similarly, we prove, also via Fourier analysis, a quantum lower bound for any function $f$ whose pure high degree is greater than or equal to $3$. These results give a large family of new exponential classical-quantum communication separations.
	\end{abstract}


    \section{Introduction}
    
    One of the main aims of the field of quantum information and quantum computation is to establish the superiority of quantum computers and quantum resources over their classical counterparts. While in some areas this superiority is based on a belief in the impossibility of classical computers solving particular tasks, e.g.\ the efficiency of Shor's algorithm~\cite{shor97} coming from the belief that there is no efficient classical factoring algorithm, in other areas like communication complexity one can establish unconditional exponential separations between classical and quantum performances.
    
    Communication complexity is a model of computation first introduced by Yao~\cite{yao1979some}. In this model, two parties (normally called Alice and Bob) hold each a piece of data and want to solve some computational task that jointly depends on their data. More specifically, if Alice holds some information $x$ and Bob holds some information $y$, they want to solve some function $f(x,y)$ or relational problem with several valid outputs for each $x$ and $y$. In order to do so, they will need to communicate between themselves, and their goal is to solve the problem with minimal communication. The protocol that Alice and Bob employ could be \emph{two-way}, where they take turns sending messages to each other; \emph{one-way}, where Alice sends a single message to Bob who then outputs the answer; or \emph{simultaneous}, where Alice and Bob each pass one message to a third party (the referee) who outputs the answer. Apart from these different types of communication settings, one is also interested in the error of a protocol when solving a communication problem: the zero-error \emph{communication complexity} is the worst-case communication of the best protocol that gives a correct output with probability $1$ for every input $(x,y)$; the bounded-error \emph{communication complexity} is the worst-case communication cost of the best protocol that gives a correct output with probability $1-\epsilon$ for every input $(x,y)$, with $\epsilon\in[0,1/2)$.
    
    An interesting extension of the original communication model is the model of \emph{quantum} communication complexity~\cite{buhrman10}, also introduced by Yao~\cite{yao1993quantum}. In this model, Alice and Bob each has a quantum computer and they exchange qubits instead of bits and/or make use of shared entanglement. The use of quantum resources can drastically reduce the amount of communication in solving some problems in comparison to the classical communication model. 
    %
    
     Exponential quantum-classical separations are known in the two-way~\cite{raz1999exponential}, one-way~\cite{bar2004exponential,gavinsky2007exponential} and simultaneous~\cite{buhrman2001quantum,doriguello19} models. Indeed, it is even known that one-way quantum communication can be exponentially more efficient than two-way classical communication~\cite{gavinsky08b,klartag11}. However, surprisingly few examples of such exponential separations are known, compared (for example) with the model of query complexity in which Shor's algorithm operates.
    
    The Hidden Matching problem~\cite{bar2004exponential} was the first problem to exhibit an exponential separation between the bounded-error classical communication complexity and the bounded-error quantum communication complexity in the one-way model. The problem can be efficiently solved by one quantum message of $\log{n}$ qubits, while any classical one-way protocol needs to send $O(\sqrt{n})$ bits to solve it. The hardness of the problem is essentially one-way: it could be efficiently solved by having Bob sent a classical message of $\log{n}$ bits to Alice. The Hidden Matching problem is a relational problem. In the same paper~\cite{bar2004exponential} the authors proposed a Boolean version of the problem, the Boolean Hidden Matching problem (which is a partial Boolean function), and conjectured that the same quantum-classical gap holds for it as well, which was later proven to be true by Gavinsky~\emph{et al.}~\cite{gavinsky2007exponential}. Generalising this separation is the focus of this work.
    
    \subsection{Hidden matching problems}
    
    Throughout the paper, $[n] := \{1,2,\dots,n\}$, $\mathbb{S}_n$ is the set of permutations of $[n]$ and given $x,y\in\{-1,1\}^n$,\addtocounter{footnote}{2}\footnote{Throughout this paper we shall use $\{-1,1\}$ instead of $\{0,1\}$ for convenience.} we denote by $x\circ y$ the Hadamard (elementwise) product of $x$ and $y$, and by $\overline{x}$ the complement of $x$, such that $x\circ \overline{x} = -1^n$.
    
    The Hidden Matching ($\operatorname{HM}^\alpha_n$) and Boolean Hidden Matching ($\operatorname{BHM}^\alpha_n$) problems are defined with respect to some $\alpha\in(0,1]$. Alice is given a string $x\in\{-1,1\}^n$ and Bob is given a sequence $M\in \mathcal{M}_{\alpha n,2}$ of $\alpha n/2$ disjoint pairs $(i_1,j_1),(i_2,j_2),\dots,(i_{\alpha n/2}, j_{\alpha n/2})\in[n]^2$. Such a sequence is called an $\alpha$-matching, and $\mathcal{M}_{\alpha n,2}$ denotes the family of all
    $\alpha$-matchings -- i.e., partial matchings of a fixed size in the complete graph on $\alpha n$ vertices. Together $x$ and $M$ induce a string $z\in\{-1,1\}^{\alpha n/2}$ defined by the parities of the $\alpha n/2$ edges, i.e., $z_\ell = x_{i_\ell} x_{j_\ell}$ for $\ell = 1,\dots,\alpha n/2$. Then the $\operatorname{HM}^\alpha_n$ and $\operatorname{BHM}^\alpha_n$ problems are defined as follows.
    \begin{HM}
        Let $n\in\mathbb{N}$ be even and $\alpha\in(0,1]$. Alice receives $x\in\{-1,1\}^n$ and Bob receives $M\in\mathcal{M}_{\alpha n,2}$. Their goal is to output a tuple $\langle i,j,b\rangle$ such that $(i,j)\in M$ and $b = x_i x_j$.
    \end{HM}
    \begin{BHM}
        Let $n\in\mathbb{N}$ be even and $\alpha\in(0,1]$. Alice receives $x\in\{-1,1\}^n$ and Bob receives $M\in\mathcal{M}_{\alpha n,2}$ and $w\in\{-1,1\}^{\alpha n/2}$. It is promised that $z\circ w = b^{\alpha n/2}$ for some $b\in\{-1,1\}$. Their goal is to output $b$.
    \end{BHM}
    Given inputs $x$ and $M$, it is clear that there are many possible correct outputs for the $\operatorname{HM}^\alpha_n$ problem ($\alpha n/2$ correct outputs, actually), making it a relational problem. On the other hand, the $\operatorname{BHM}^\alpha_n$ is a partial Boolean function due to the promise statement.

    Bar-Yossef \emph{et al.}~\cite{bar2004exponential} gave a simple quantum protocol to solve the $\operatorname{HM}^\alpha_n$ problem with just $O(\log{n})$ qubits of communication for any $\alpha$, while proving that any classical protocol needs to communicate at least $\Omega(\sqrt{n})$ bits in order to solve it. Similarly with the $\operatorname{BHM}^\alpha_n$ problem, Gavinsky \emph{et al.}~\cite{gavinsky2007exponential} demonstrated the same exponential classical-quantum communication gap for any $\alpha \le 1/2$ (note that their definition of $\alpha$ differs from ours by a factor of $2$). As $\operatorname{HM}^\alpha_n$ is at least as difficult as $\operatorname{BHM}^\alpha_n$, their result implies the same lower bound for $\operatorname{HM}^\alpha_n$. The approach taken by Gavinsky~\emph{et al.} in proving the classical lower bound is particularly interesting in that it uses the Fourier coefficients inequality of Kahn, Kalai, and Linial~\cite{kahn1988influence}, which is proven via the Bonami-Beckner inequality~\cite{bonami1970etude,beckner1975inequalities}. We also mention that Fourier analysis had been previously used in communication complexity by Raz~\cite{raz1995fourier} and Klauck~\cite{klauck2001lower}.
    
    A slightly weaker separation ($O(\log n)$ vs.\ $\Omega(n^{7/16})$) for a closely related problem was shown in~\cite{montanaro11} using similar techniques. The $\operatorname{BHM}^\alpha_n$ problem was generalised by Verbin and Yu~\cite{verbin2011streaming} to a problem that they named Boolean Hidden Hypermatching ($\operatorname{BHH}^{\alpha}_{t,n}$). In this problem, instead of having the bits from Alice matched in pairs, they are now matched in tuples of $t$ elements. In other words, a bit from the final string $z$ is obtained by XORing $t$ bits from Alice's string. More precisely, Alice is given a string $x\in\{-1,1\}^n$ and Bob is given a sequence $M\in \mathcal{M}_{\alpha n,t}$ of $\alpha n/t$ disjoint tuples $(M_{1,1},\dots,M_{1,t}),\dots,(M_{\alpha n/t,1},\dots,M_{\alpha n/t,t})\in[n]^t$ called an $\alpha$-hypermatching, where $\mathcal{M}_{\alpha n,t}$ denotes the family of all such $\alpha$-hypermatchings. Both $x$ and $M$ induce a string $z\in\{-1,1\}^{\alpha n/t}$ defined by the parities of the $\alpha n/t$ hyperedges, i.e., $z_\ell = \prod_{j=1}^t x_{M_{\ell,j}}$ for $\ell = 1,\dots, \alpha n/t$. The $\operatorname{BHH}^{\alpha}_{t,n}$ problem is defined as follows.
    \begin{BHH}
        Let $n,t\in\mathbb{N}$ be such that $t|n$ and $\alpha\in(0,1]$. Alice receives $x\in\{-1,1\}^n$ and Bob receives $M\in\mathcal{M}_{\alpha n,t}$ and $w\in\{-1,1\}^{\alpha n/t}$. It is promised that $z\circ w = b^{\alpha n/t}$ for some $b\in\{-1,1\}$. Their goal is to output $b$.
    \end{BHH}
    Verbin and Yu proved a classical lower bound of $\Omega(n^{1-1/t})$ communication for every bounded-error one-way protocol, showing the increasing hardness of the problem with $t$, as one should expect since the $\operatorname{BHH}^\alpha_{t,n}$ problem can be reduced from the $\operatorname{BHM}^\alpha_n$ problem (we will show how this is done in detail later). The authors subsequently used this problem to prove various lower bounds on the space required of streaming algorithms (algorithms that read the input from left to right, use a small amount of space, and approximate some function of the input). However, no efficient quantum protocol was proposed for solving the $\operatorname{BHH}^\alpha_{t,n}$ problem for $t>2$. It was only later that Shi, Wu and Yu~\cite{shi2012limits} showed that such efficient quantum protocols do not exist. More specifically, they proved a quantum lower bound of $\Omega(n^{1-2/t})$ communication for every bounded-error one-way protocol for the $\operatorname{BHH}^\alpha_{t,n}$ problem. Their proof is similar to the ones used in the classical lower bound, the difference lying in the use of Fourier analysis of \emph{matrix-valued} functions and the matrix-valued Hypercontractive Inequality of Ben-Aroya, Regev and de Wolf~\cite{ben2008hypercontractive}.
    
    The original lower bound of Verbin and Yu assumes $\alpha=1$, unlike the lower bound of~\cite{gavinsky2007exponential}, where $\alpha \leq 1/2$. However, their lower bound requires $n/t$ to be even, otherwise Alice can just send the parity of her bit-string. (The result of~\cite{gavinsky2007exponential} can be extended to hold for any $\alpha < 1$ fairly straightforwardly, but achieving a strong lower bound for $\alpha = 1$ requires some more work.)
    
    \subsection{Our Results}
    
    This paper focuses on the study of a broad generalisation of the $\operatorname{BHH}_{t,n}^\alpha$ problem.
    In the (Boolean) Hidden Matching and Boolean Hidden Hypermatching problems, the task Alice and Bob want to solve can be viewed as rearranging Alice's data according to some permutation that Bob holds, and `compressing' the data to a final bit-string by applying some Boolean function to the bits. Then Alice and Bob's goal is to determine some information about this final bit-string. The way this compression was originally done was via the Parity function, but, apart from the obvious reason that Parity gives the desired classical-quantum communication gap and, less obviously, leads to a clear proof, there is no particular need to restrict to this function in order to arrive at the final bit-string. This observation leads to a generalisation of the Boolean Hidden Hypermatching problem, which we named the $f$-Boolean Hidden Partition ($f\operatorname{-BHP}^{\alpha}_{t,n}$) problem, where $f:\{-1,1\}^t\to\{-1,1\}$ is the Boolean function used to compress Alice's bits.
    
    Given $y\in\{-1,1\}^n$, we define by $y^{(j)} = (y_{(j-1)t+1},y_{(j-1)t+2},\dots,y_{jt})\in\{-1,1\}^t$ the $j$-th block of size $t$ from $y$, with $t|n$ and $j=1,\dots,n/t$. The $f$-Boolean Hidden Partition problem is defined as follows. Alice is given a bit-string $x\in\{-1,1\}^n$, and Bob is given a permutation $\sigma\in \mathbb{S}_n$ and a bit-string $w\in\{-1,1\}^{\alpha n/t}$, where $\alpha\in(0,1]$ is fixed and $t|n$. Given a Boolean function $f:\{-1,1\}^t\to\{-1,1\}$, we can define the map $B_f:\{-1,1\}^n\times\mathbb{S}_n\to \{-1,1\}^{\alpha n/t}$ by $B_f(x,\sigma) = \big(f(\sigma(x)^{(1)}),\dots,f(\sigma(x)^{(\alpha n/t)})\big)$, where $\sigma(x)_i = x_{\sigma^{-1}(i)}$. Hence $x$ and $\sigma$ induce a bit-string given by $B_f(x,\sigma)$, each of whose bits is obtained by applying $f$ to a block of the permuted bit-string $\sigma(x)$. The $f\operatorname{-BHP}^{\alpha}_{t,n}$ problem can be defined as follows.
    \begin{BHP}
        Let $n,t\in\mathbb{N}$ be such that $t|n$ and $\alpha\in(0,1]$. Alice receives $x\in\{-1,1\}^n$ and Bob receives $\sigma\in \mathbb{S}_n$ and $w\in\{-1,1\}^{\alpha n/t}$. It is promised that there exists $b\in\{-1,1\}$ such that $B_f(x,\sigma)\circ w = b^{\alpha n/t}$. The problem is to output $b$.
    \end{BHP}
    
    The adoption of the word `Partition' instead of `(Hyper)Matching' from previous works comes from our decision to view the problem in terms of a hidden partition that Bob holds, instead of an $\alpha$-(Hyper)Matching. Bob shuffles Alice's data according to some permutation, and then just partitions the resulting data in adjacent blocks of size $t$ and uses $f$ to get the final bit-string. Obviously both views are equivalent, but we think that the permutation approach eases the analysis of the problem.
    
    Our aim is to study the $f$-Boolean Hidden Partition problem in terms of the function $f$. It should be clear that for some functions the problem is hard to solve classically, e.g.\ when $f$ is the Parity function and we recover the usual Boolean Hidden Hypermatching problem. On the other hand, for some functions it becomes easily solvable, e.g.\ when $f$ is the AND function, since Alice needs only to send the position of any $1$ in her string (thinking of $1$ as $0$ when using $\{-1,1\}$). We would like to characterize for which functions the problem can be efficiently solved classically, i.e., with $O(\log{n})$ bits of communication, and for which functions it is hard to solve classically, i.e., requires $\Omega(n^{a})$ bits of communication for some $a\in(0,1]$. And the same question applies to quantum communication complexity: we would like to determine for which functions the problem admits or not an efficient quantum communication protocol. Given this characterization, we can check for which functions there is an exponential classical-quantum communication gap.
    
    We conjecture that the whole $f\operatorname{-BHP}^{\alpha}_{t,n}$ problem can be characterized mainly by the \emph{sign-degree} of the function $f$, and we give substantial evidence for such conjecture. A polynomial $p:\{-1,1\}^t\to \mathbb{R}$ is said to \emph{sign-represent} $f$ if $f(x) = \operatorname{sgn}(p(x))$. If $|p(x)| \leq 1$ for all $x$, we say that $p$ is \emph{normalized}. The \emph{bias} of a normalized polynomial $p$ is defined as $\beta = \min_x|p(x)|$. The sign-degree ($\operatorname{sdeg}(f)$) of $f$ is the minimum degree of polynomials that sign-represent it. In Section~\ref{sec:sec3} we give upper bounds on the classical and quantum communication complexity of the $f$-Boolean Hidden Partition problem based on the sign-degree.
    \begin{customthm}{\ref{thr:thr2.1}+\ref{thr:thr3.2}}
        Let $f:\{-1,1\}^t\to\{-1,1\}$ be a Boolean function. If $\operatorname{sdeg}(f) \leq 1$, then there is a bounded-error classical protocol that solves $f\operatorname{-BHP}^{\alpha}_{t,n}$ with error probability $\epsilon$ and $O\big((\frac{t}{\alpha\beta})^2\log\frac{1}{\epsilon}\log{n}\big)$ bits of communication. If $\operatorname{sdeg}(f) \leq 2$, then there is a bounded-error quantum protocol that solves $f\operatorname{-BHP}^{\alpha}_{t,n}$ with error probability $\epsilon$ and $O\big((\frac{t}{\alpha\beta})^2\log\frac{1}{\epsilon}\log{n}\big)$ qubits of communication. In both these results, $\beta$ is the bias of any normalized polynomial of degree $\operatorname{sdeg}(f)$ that sign-represents $f$.
    \end{customthm}
    Note that the bias $\beta$ can be very small, but can also be lower-bounded in terms of $t$ and $\operatorname{sdeg}(f)$: indeed, it is shown in~\cite{buhrman07} that $\beta$ is lower-bounded by $t^{-O(t^{\operatorname{sdeg}(f)})}$. In this work we will usually assume that $t=O(1)$, so $\beta = \Omega(1)$. We assume throughout that Alice and Bob do not have access to shared randomness or entanglement. The classical complexity in the above theorem can actually be improved to an additive dependence on $\log n$ via applying Newman's Theorem~\cite{newman91} to a protocol with shared randomness, but at the expense of making the protocol less intuitive.
    
    The classical upper bound stated above comes from the observation that, if $f$ has a sign-representing polynomial $p$ of degree 1, it is possible to determine whether $f(z)=1$ with probability $>1/2$ by only evaluating $f$ on one uniformly random bit of $z$, by writing down a probabilistic procedure whose expectation on $z$ mimics $p(z)$. So Alice sends a few uniformly random bits to Bob, who matches them to blocks in his partition, and evaluates $f$ on the corresponding blocks with success probability $>1/2$ for each block. Only a few repetitions are required to determine whether $f(x)=w$ or $f(x) = \overline{w}$ with high probability.
    
    On the other hand, to obtain the quantum upper bound we use the idea of \emph{block-multilinear} polynomials from~\cite{aaronson2018forrelation, aaronson2016polynomials}, and some auxiliary results also from~\cite{aaronson2016polynomials}. The idea is that Alice sends a superposition of her bits, and Bob, after collapsing the state onto one of the blocks from his partition (say block $j$), applies a controlled unitary operator that describes a block-multilinear polynomial $\widetilde{p}$ of degree $2$, which is produced from a sign-representing polynomial $p$ for $f$ of degree~2. A Hadamard test is used to return an output with probability depending (roughly speaking) on $\widetilde{p}(\sigma(x)^{(j)},\sigma(x)^{(j)})$, which in turn is equal to $p(\sigma(x)^{(j)})$ according to~\cite[Theorem~4]{aaronson2016polynomials}. The Hadamard test then outputs $0$ with probability greater than $1/2$ if $f(\sigma(x)^{(j)}) = 1$ and $1$ with probability greater than $1/2$ if $f(\sigma(x)^{(j)}) = -1$.
    
    We note that our quantum upper bound and the protocol behind it generalises a result from Wehner and de Wolf~\cite[Theorem~2]{MR2184730}, where they showed a quantum algorithm to compute $f(x_0,x_1)$, where $f:\{-1,1\}^{2b}\to\{-1,1\}$, with success probability greater than $1/2$ using just one copy of $\frac{1}{\sqrt{2}}(|x_0,0\rangle + |x_1,1\rangle)$. Here we are able to compute $f(x_0,\dots,x_{t-1})$, where $f:\{-1,1\}^t\to\{-1,1\}$, with probability greater than $1/2$ with just one copy of $\frac{1}{\sqrt{t}}\sum_{j=0}^{t-1}|x_j,j\rangle$ if $\operatorname{sdeg}(f)\leq 2$.
    
    We remark that both of these protocols actually solve a natural generalisation of the Hidden Matching problem~\cite{bar2004exponential} (i.e., they output the result of evaluating $f(x^{(j)})$ for Bob's block $j$, where $j$ is arbitrary), which is at least as hard as the $f$-Boolean Hidden Partition problem. However, unlike the Hidden Matching problem, the output is not correct with certainty, but only with probability strictly greater than $1/2$.
    
    In Sections~\ref{sec:sec4},~\ref{sec:sec5} and~\ref{sec:sec6.1} we prove classical and quantum lower bounds. In Section \ref{sec:sec4} we reduce the Boolean Hidden Matching problem to the $f$-Boolean Hidden Partition problem and prove that for almost all symmetric Boolean function $f$ with $\operatorname{sdeg}(f)\geq 2$ the $f\operatorname{-BHP}^{\alpha}_{t,n}$ problem requires at least $\Omega(\sqrt{n})$ bits of communication. The only functions for which the reduction does not work are the Not All Equal functions on an odd number of bits, i.e., $\operatorname{NAE}:\{-1,1\}^t\to\{-1,1\}$ defined by $\operatorname{NAE}(x) = -1$ if $|x|\in\{0,t\}$ and $\operatorname{NAE}(x) = 1$ otherwise, with $t$ odd.
    \begin{customthm}{\ref{thr:thr4}}
        Let $f:\{-1,1\}^t\to\{-1,1\}$ be a symmetric Boolean function with $\operatorname{sdeg}(f) \geq 2$. If $f$ is not the NAE function on an odd number of bits, then any bounded-error classical communication protocol for solving the $f\operatorname{-BHP}^{\alpha}_{t,n}$ problem needs to communicate at least $\Omega(\sqrt{n/(\alpha t)})$ bits.
    \end{customthm}
    Finally, in Sections~\ref{sec:sec5} and~\ref{sec:sec6.1} we generalise the Fourier analysis methods from \cite{gavinsky2007exponential, verbin2011streaming, shi2012limits} to prove a partial result on the hardness of the $f\operatorname{-BHP}^{\alpha}_{t,n}$ problem, both classically and quantumly. Ideally we would like to prove that any bounded-error classical and quantum protocols would need to communicate $\Omega(n^{1-1/d})$ bits and $\Omega(n^{1-2/d})$ qubits, respectively, where $\operatorname{sdeg}(f) = d$. What we obtained is this result but with $d$ being the \emph{pure high degree} of $f$. A Boolean function $f$ is said to have pure high degree ($\operatorname{phdeg}(f)$) $d$ if $\widehat{f}(S) = 0$ for all $|S| = 0,1,\dots,d-1$, where $\widehat{f}(S) = \frac{1}{2^n}\sum_{x\in\{-1,1\}^n} f(x)\chi_S(x)$ is the Fourier transform of $f$ and $\chi_S(x) = \prod_{i\in S} x_i$, with $S\subseteq [n]$, is a character function. It is possible to prove that $\operatorname{phdeg}(f) \leq \operatorname{sdeg}(f)$~\cite{sherstov2011pattern} (see also~\cite[Theorem~1]{bun2015dual}), so our result is a step towards proving a lower bound for all functions with sign-degree $\geq 2$.
    \begin{customthm}{\ref{thr:thr5.1} + \ref{thr:thr6.1}}
        Let $f:\{-1,1\}^t\to\{-1,1\}$ be a Boolean function. If $\operatorname{phdeg}(f) = d \geq 2$, then, for sufficiently small $\alpha > 0$ that does not depend on $n$, any bounded-error classical communication protocol for solving the $f\operatorname{-BHP}^{\alpha}_{t,n}$ problem needs to communicate at least $\Omega(n^{1-1/d})$ bits. If $\operatorname{phdeg}(f) = d \geq 3$, then, for sufficiently small $\alpha > 0$ that does not depend on $n$, any bounded-error quantum communication protocol for solving the $f\operatorname{-BHP}^{\alpha}_{t,n}$ problem needs to communicate at least $\Omega(n^{1-2/d})$ qubits. 
    \end{customthm}
    The classical proof in Section~\ref{sec:sec5} follows the general idea from~\cite{gavinsky2007exponential,verbin2011streaming}. First, we apply Yao's minimax principle~\cite{yao1977probabilistic}, which says that it suffices to prove a lower bound for a \emph{deterministic} protocol under a hard probability distribution on Alice and Bob's inputs. Alice sends a message to Bob. If the length of the message sent is $c$, then the inputs for which Alice could have sent that specific message define a set $A$ of about $2^{n-c}$ $x$'s. From Bob's perspective, he knows that the random variable $X$ corresponding to Alice's bit-string is uniformly distributed in a set $A$ and he knows his permutation $\sigma$, hence his knowledge of the random variable $B_f(X,\sigma)$ is described by the distributions
    \begin{align*}
        p_\sigma(z) = \frac{|\{x\in A|B_f(x,\sigma) = z\}|}{|A|} ~\text{and}~ q_\sigma(z) = \frac{|\{x\in A| B_f(x,\sigma) = \overline{z}\}|}{|A|}.
    \end{align*}
    It is well known that the best success probability for distinguishing two distributions $q_1$ and $q_2$ with one sample is $1/2 + \|q_1 - q_2\|_{\operatorname{tvd}}/4$, where $\|q_1 - q_2\|_{\operatorname{tvd}} := \sum_i |q_1(i)-q_2(i)|$. Therefore the bias of the protocol, i.e., the protocol's successful probability minus a half, is equal to $\|p_\sigma - q_\sigma\|_{\operatorname{tvd}}/4$. We show that, if the amount of communication from Alice to Bob is not large enough, then $\|p_\sigma - q_\sigma\|_{\operatorname{tvd}}$ is small, and thus Bob cannot differentiate between $p_\sigma$ and $q_\sigma$. Upper-bounding the total variation distance is done via Fourier analysis, using the inequality of Kahn, Kalai and Linial~\cite{kahn1988influence}.

    The quantum proof in Section~\ref{sec:sec6.1} follows the same idea from~\cite{shi2012limits}, but the second half of the proof was modified by borrowing ideas from the classical proofs~\cite{gavinsky2007exponential,verbin2011streaming}. Yao's minimax principle is still applied, and the best strategy for Bob in determining $b$ conditioned on his input $(\sigma, w)$ is no more than the chance to distinguish between the two statistical ensembles of Alice's messages, where a message corresponds to a quantum state $\rho_x$ encoding Alice's string $x$, selected according to $b$. It is known that any protocol that tries to distinguish two quantum states $\rho_0$ and $\rho_1$ appearing with probability $p$ and $1-p$, respectively, by a POVM has bias at most $\|p\rho_0 - (1-p)\rho_1\|_{\operatorname{tr}}/2$~\cite{helstrom1976quantum}. The bias is then upper-bounded by using Fourier analysis of matrix-valued functions, in particular by the matrix-valued hypercontractive inequality of Ben-Aroya, Regev and de Wolf~\cite{ben2008hypercontractive}. 
    
    The difference between the classical and quantum lower bound proofs was considerably reduced in our paper, e.g.\ the quantum lower bound proof now borrows the idea from~\cite{gavinsky2007exponential,verbin2011streaming} of splitting a sum bounding the bias in two parts instead of performing it at once as in~\cite{shi2012limits}, which actually leads to a better $\alpha$ dependence. Still some differences persist. Apart from the obvious generalisation of Fourier analysis to matrix-valued functions, the Fourier analysis in the quantum lower bound proof is performed directly on the encoding messages and not on the pre-images of a fixed encoding message, since there is no clear quantum analogue of conditioning on a message.
    The main technical difficulty we faced compared to~\cite{gavinsky2007exponential,verbin2011streaming} is that the Fourier coefficients of Bob's distributions $p_{\sigma}(z)$ and $q_\sigma(z)$ are not nicely related to just one Fourier coefficient of the characteristic function of $A$ any more, but instead to a more complicated sum of many coefficients. This requires us to carefully bound various combinatorial terms occurring in the proof and to use our freedom to choose $\alpha$ less than a (potentially small) constant depending on the Boolean function $f$.
    
    In Section~\ref{sec:sec6} we analyse the limitations of our techniques and show that under the uniform distribution, which was used as the `hard' distribution during the proof of Theorem~\ref{thr:thr5.1}, we cannot obtain a lower bound depending on the sign-degree instead of the pure high degree.
    
    We finally remark that the one-way communication complexity separations we found can easily be used to obtain corresponding separations in the streaming model, similarly to~\cite{gavinsky2007exponential,verbin2011streaming}.
    
    \section{Classical and Quantum Upper Bounds}\label{sec:sec3}
    
    The sign-representing polynomial $p$ allows us to build efficient classical and quantum communication protocols depending on $\operatorname{sdeg}(f)$. We shall show that there is an efficient $O(\log{n})$-bit classical communication protocol for solving the $f$-$\operatorname{BHP}^{\alpha}_{t,n}$ problem if $\operatorname{sdeg}(f) \leq 1$. On the other hand, we shall show that there is an efficient $O(\log{n})$-qubit quantum communication protocol for solving the $f$-$\operatorname{BHP}^{\alpha}_{t,n}$ problem if $\operatorname{sdeg}(f) \leq 2$.
    
    Intuitively, the contrast between $\operatorname{sdeg}(f) \leq 1$ for the classical protocols and $\operatorname{sdeg}(f) \leq 2$ for the quantum protocols comes from the nature of probability distributions in each case. One wants to relate the probability of outputting the right answer with the sign-representing polynomial $p$: if $p(x) > 0$, we would like to output $1$ with high probability, and if $p(x) < 0$, we would like to output $-1$ with high probability. Classically, this probability distribution can only depend linearly on the bits of $x$, but quantumly, since this probability distribution arises from the square of a quantum amplitude, it can have a quadratic dependence on the bits of $x$.
    
    \subsection{Classical Upper Bound}
    
    Consider the $f$-$\operatorname{BHP}^{\alpha}_{t,n}$ problem for $f:\{-1,1\}^t\to\{-1,1\}$ with $\operatorname{sdeg}(f) \leq 1$. Let $p:\{-1,1\}^t \to [-1,1]$ be a normalized sign-representing polynomial for $f$. Hence we can write
	\begin{align*}
		p(x) = \alpha_0 + \sum_{i=1}^t \alpha_i x_i
	\end{align*}
	with $(\alpha_i)_{i=0}^t \in \mathbb{R}$. Let $\beta = \min_x |p(x)|$ be the bias of $p$.
	
	In the following, denote by $R^1_\epsilon(\mathcal{P})$ and $Q^1_\epsilon(\mathcal{P})$ the classical and quantum communication cost of the protocol $\mathcal{P}$ in bits and qubits, respectively, and denote by $R^1_\epsilon(f) = \min_{\mathcal{P}}R^1_\epsilon(\mathcal{P})$ and $Q^1_\epsilon(f) = \min_{\mathcal{P}}Q^1_\epsilon(\mathcal{P})$ the minimum classical and quantum communication cost, respectively, over all one-way protocols $\mathcal{P}$ without shared randomness that solve a communication problem $f$ with failure probability $0 < \epsilon < 1/2$. Define $R^1(f) := R^1_{1/3}(f)$ and $Q^1(f) := Q^1_{1/3}(f)$ (or any constant bounded away from $0$ and $1/2$).
	\begin{theorem}\label{thr:thr2.1}
	    $R^1_\epsilon(f\operatorname{-BHP}^{\alpha}_{t,n}) = O\left((\frac{t}{\alpha\beta})^2\log\frac{1}{\epsilon}\log{n}\right)$ if $\operatorname{sdeg}(f) \leq 1$, where $\beta$ is the bias of any normalized sign-representing polynomial for $f$ with degree $\leq 1$.
	\end{theorem}
    \begin{proof}
        Consider the following protocol: Alice picks $m = O\big((\frac{t}{\alpha\beta})^2\log\frac{1}{\epsilon}\big)$ bits from $x$ uniformly at random (with replacement) and sends them to Bob, together with their indices. Let $I$ and $\{x_i\}_{i\in I}$ be the indices and bitvalues sent, respectively. Let $j(i) = \lceil \sigma(i)/t\rceil$ and $k(i) \equiv \sigma(i)~(\text{mod}~t)$ for all $i\in I$, where $\sigma\in \mathbb{S}_n$ is Bob's permutation. Define the random variable
        \begin{align*}
            X(i) := \begin{cases}
                (\alpha_{k(i)} x_i + \alpha_0/t) w_{j(i)} &~\text{if}~ \sigma(i)\in[\alpha n/t],\\
                0 &~\text{if}~ \sigma(i)\notin[\alpha n/t],
            \end{cases}
        \end{align*}
        where $\alpha_0$ and $\alpha_k$ are the zeroth order and $x_k$'s coefficients, respectively, from the sign-representing polynomial $p$, and define $X := \sum_{i\in I} X(i)$. Bob then computes $\sgn(X)$. If the sign is $1$, then he outputs $B_f(x,\sigma) = w$, and if the sign is $-1$, then he outputs $B_f(x,\sigma) = \overline{w}$.
        
        To see why the protocol works, we calculate the expected value of the random variable $X$.
    \begin{align*}
	    \displaybreak[0]\mathbb{E}[X] &= m \cdot\mathbb{E}_i[X(i)]\\\displaybreak[0]
	    &= \alpha m\cdot\mathbb{E}_{i} [(\alpha_{k(i)}x_i + \alpha_0/t) w_{j(i)}]\\\displaybreak[0]
	    &= \alpha m\cdot\mathbb{E}_j\big[\mathbb{E}_k[\alpha_k\sigma(x)^{(j)}_k + \alpha_0/t]w_j\big]\\\displaybreak[0]
	    &= \alpha m\cdot\mathbb{E}_j\left[\frac{p(\sigma(x)^{(j)})}{t}w_j\right]\\\displaybreak[0]
	    &= \alpha m\frac{t}{n}\sum_{j=1}^{n/t}\frac{p(\sigma(x)^{(j)})}{t}w_j\\\displaybreak[0]
	    &= \frac{\alpha m}{n}\left[\sum_{j:w_j=1}p(\sigma(x)^{(j)}) - \sum_{j:w_j=-1}p(\sigma(x)^{(j)})\right].
	\end{align*}
	If $f(\sigma(x)^{(j)}) = w_j$, then $w_j = 1 \implies p(\sigma(x)^{(j)}) \geq \beta > 0$ and $w_j = -1 \implies p(\sigma(x)^{(j)}) \leq -\beta < 0 $. Therefore
	\begin{align*}
	    \mathbb{E}[X] &\geq \frac{\alpha m}{n}\left[\sum_{j:w_j=1}\beta - \sum_{j:w_j=-1}-\beta\right] = \frac{\alpha \beta}{t}m.
	\end{align*}
	If, on the other hand, $f(\sigma(x)^{(j)}) = -w_j$, then $w_j = 1 \implies p(\sigma(x)^{(j)}) \leq -\beta < 0$ and $w_j = -1 \implies p(\sigma(x)^{(j)}) \geq \beta > 0$. Therefore
	\begin{align*}
	    \mathbb{E}[X] &\leq \frac{\alpha m}{n}\left[\sum_{j:w_j=1}-\beta - \sum_{j:w_j=-1}\beta\right] = -\frac{\alpha\beta}{t}m.
	\end{align*}
    By using a Chernoff bound~\cite[Theorem~1.1]{dubhashi09} of the type $\operatorname{Pr}[X > \mathbb{E}[X] + u], \operatorname{Pr}[X < \mathbb{E}[X] - u] \leq e^{-2u^2/m}$ with $u > 0$ and setting $u = \pm \mathbb{E}[X] > 0$, we can make 
    \begin{align*}
        \operatorname{Pr}[X > 0 | B_f(x,\sigma) = \overline{w}], ~\operatorname{Pr}[X < 0 | B_f(x,\sigma) = w] \leq \epsilon
    \end{align*}
    by taking $m = O\big((\frac{t}{\alpha\beta})^2\log\frac{1}{\epsilon}\big)$. Therefore Alice and Bob can decide if $B_f(x,\sigma) = w$ or $B_f(x,\sigma) = \overline{w}$ with error probability $\epsilon$ and $O\big((\frac{t}{\alpha\beta})^2\log\frac{1}{\epsilon}\log{n}\big)$ bits of communication.
    \end{proof}

    \subsection{Quantum Upper Bound}
    
    Consider the $f$-$\operatorname{BHP}^{\alpha}_{t,n}$ problem for $f:\{-1,1\}^t\to\{-1,1\}$ with $\operatorname{sdeg}(f) = 2$. Let $p:\{-1,1\}^t \to [-1,1]$ be a normalized sign-representing polynomial for $f$. Let $\beta = \min_x |p(x)|$ be again the bias of $p$. In the following, define $\widetilde{x} = (1,x_1,\dots,x_t)$.
    
    In order to obtain our upper bound, we borrow the idea of \emph{block-multilinear} polynomials from~\cite{aaronson2018forrelation, aaronson2016polynomials}, which are also known as multilinear forms. We say that a polynomial $q$ of degree $k$ is block-multilinear if its variables $x_1,\dots,x_N$ can be partitioned into $k$ blocks $R_1,\dots,R_k$, such that every monomial of $q$ contains exactly one variable from each block. As a special case, a block-multilinear polynomial $q$ of degree $2$ can be written as
    \begin{align*}
        q(x_1,\dots,x_n,y_1,\dots,y_m) = \sum_{\substack{i\in[n]\\j\in[m]}} a_{ij}x_iy_j
    \end{align*}
    with variables in the first block labeled as $x_1,\dots,x_n$ and the variables in the second block labeled as $y_1,\dots,y_m$. Defining the matrix $A = (a_{ij})_{i\in[n],j\in[m]}$, then
    \begin{align*}
        q(x,y) = x^TAy
    \end{align*}
    for all $x\in\mathbb{R}^n$ and $y\in\mathbb{R}^m$. We say that $q$ is \emph{bounded} if $|q(x,y)| \leq 1$ for all $x\in\{-1,1\}^n,y\in\{-1,1\}^m$. This translates to
    \begin{align*}
        \max_{\substack{x\in\{-1,1\}^n\\y\in\{-1,1\}^m}} \left|\sum_{\substack{i\in[n]\\j\in[m]}} a_{ij}x_iy_j\right| \leq 1,
    \end{align*}
    i.e., $\|A\|_{\infty\to1} \leq 1$. More generally, in the following, given a complex matrix $M$, we define $\|M\|_{p\to q} := \sup_{x\neq 0}\|Mx\|_q/\|x\|_p$ and $\|M\| := \|M\|_{2\to 2}$ is its spectral norm.
    
    We shall also use the following results (a similar version of Theorem~\ref{thr:thr3.1} was proven in~\cite{MR3540825}). 
    \begin{lemma}[{\cite[Lemma 7]{aaronson2016polynomials}}]\label{lem:lem1}
        Given an $m\times m$ complex matrix $M$, there exists a unitary $U$ (on a possibly larger space with basis $|1\rangle, \dots, |k\rangle$ for some $k\geq m$) such that, for any unit vector $|y\rangle = \sum_{i=1}^m \alpha_i|i\rangle$, $U|y\rangle = \frac{M|y\rangle}{\|M\|} + |\phi\rangle$, where $|\phi\rangle$ consists of basis states $|i\rangle$, $i > m$ only.
	\end{lemma}
    \begin{theorem}[{\cite[Theorem 4]{aaronson2016polynomials}}]\label{thr:thr3.1}
	Let $p:\{-1,1\}^t \to [-1,1]$ be a sign-representing polynomial for $f$ with $\operatorname{sdeg}(f) = 2$. Then there is a block-multilinear polynomial $\widetilde{p}: \mathbb{R}^{2(t+1)}\to\mathbb{R}$ such that $\widetilde{p}(\widetilde{x},\widetilde{x}) = p(x)$ for any $x\in\{-1,1\}^t$, and $|\widetilde{p}(y)|\leq 3$ for any $y\in \{-1,1\}^{2(t+1)}$. 
	\end{theorem}
    Let $\widetilde{p}: \mathbb{R}^{2(t+1)}\to\mathbb{R}$ be the block-multilinear polynomial of degree $2$ obtained from the sign-representing polynomial $p$ of $f$ according to Theorem~\ref{thr:thr3.1}. It can be written as
	\begin{align}\label{eq:eq3.1}
		\widetilde{p}(x,y) = \sum_{i,j\in[t+1]} a_{ij}x_iy_j = x^T Ay,
	\end{align}
	where $A = (a_{ij})_{i,j\in[t+1]}$.
    
    With these in hands, we present our upper bound.
    \begin{theorem}\label{thr:thr3.2}
       $Q^1_\epsilon(f\operatorname{-BHP}^{\alpha}_{t,n}) = O\left((\frac{t}{\alpha\beta})^2\log\frac{1}{\epsilon}\log{n}\right)$ if $\operatorname{sdeg}(f) \leq 2$, where $\beta$ is the bias of any normalized sign-representing polynomial for $f$ with degree $\leq 2$.
    \end{theorem}
    \begin{proof}
        The case $\operatorname{sdeg}(f) \leq 1$ follows from Theorem~\ref{thr:thr2.1}. Assume then that $\operatorname{sdeg}(f) = 2$ and consider the following protocol: Alice sends to Bob $m = O\big((\frac{t}{\alpha\beta})^2\log\frac{1}{\epsilon}\big)$ copies of the $O(\log n)$-qubit quantum state
        \begin{align*}
	 	    |\psi_A\rangle = \frac{1}{\sqrt{n+n/t}}\left(\sum_{i=1}^{n}x_i|i\rangle + \sum_{i=1}^{n/t}|n+i\rangle\right).
	    \end{align*}
	 Bob measures each of them by using the POVM
	  \begin{align*}
	 	\left\{|n+j\rangle\langle n+j| + \sum_{i=(j-1)t+1}^{jt}|\sigma^{-1}(i)\rangle\langle \sigma^{-1}(i)|\right\}_{j\in[n/t]}, 
	 \end{align*}
	 where $\sigma\in \mathbb{S}_n$ is his permutation, and attaches a qubit in the state $|+\rangle$ to each of the resulting states. Let $I\in [n/t]^m$ be the sequence of indices from his measurements. Then his state is
	 \begin{align*}
	 	|\psi_B\rangle = \bigotimes_{j\in I}|+\rangle|\psi^{(j)} \rangle,
	 \end{align*}
	 where
	 \begin{align*}
	 	|\psi^{(j)}\rangle = \frac{1}{\sqrt{t+1}}\left(|n+j\rangle + \sum_{i=(j-1)t+1}^{jt}x_{\sigma^{-1}(i)}|\sigma^{-1}(i)\rangle\right).
	 \end{align*}

	 Let $A$ be the $(t+1)\times(t+1)$ matrix from the representation of $\widetilde{p}$ according to Eq.~(\ref{eq:eq3.1}). Lemma~\ref{lem:lem1} guarantees the existence of a unitary $U_j$ such that $U_j|\psi^{(j)}\rangle = \frac{A|\psi^{(j)}\rangle}{\|A\|} + |\phi^{(j)}\rangle$, with $\langle \phi^{(j)}|\psi^{(j)}\rangle = 0$. Bob then applies a controlled $U_j$ gate onto each $|+\rangle_j|\psi^{(j)} \rangle$ to obtain
	 \begin{align*}
	     \bigotimes_{j\in I}CU_j|\psi_B\rangle = \bigotimes_{j\in I}\left(\frac{1}{\sqrt{2}}|0\rangle|\psi^{(j)}\rangle + \frac{1}{\sqrt{2}}|1\rangle U_j|\psi^{(j)}\rangle\right)
	 \end{align*}
	 and then performs a Hadamard gate on the first qubit of each of the subsystems $I$ and measures them on the computational basis. Let $m_j\in\{0,1\}$ be the result of the measurement for block $j\in I$. Define the random variable 
    \begin{align*}
        X(j) := \begin{cases}
            (-1)^{m_j}w_j &~\text{if}~ j\in[\alpha n/t],\\
            0 &~\text{if}~ j\notin[\alpha n/t],
        \end{cases}
    \end{align*}
    and define $X := \sum_{j\in I}X(j)$. Bob then computes $\operatorname{sgn}(X)$: if $\operatorname{sgn}(X) > 0$, he outputs that $B_f(x,\sigma) = w$, and if $\operatorname{sgn}(X) < 0$, he outputs that $B_f(x,\sigma) = \overline{w}$.
	 
	 To see why the protocol works, first note that the probability of measuring $0$ is
	 \begin{align*}
	 	\frac{1}{2} + \frac{1}{2}\langle \psi^{(j)}|U|\psi^{(j)}\rangle = \frac{1}{2} + \frac{\langle \psi^{(j)}|A|\psi^{(j)}\rangle}{2\|A\|} = \frac{1}{2} + \frac{\widetilde{p}(\widetilde{\sigma(x)^{(j)}},\widetilde{\sigma(x)^{(j)}})}{2\|A\|(t+1)} = \frac{1}{2} + \frac{p(\sigma(x)^{(j)})}{2\|A\|(t+1)}.
	 \end{align*}
	 The rest of the argument is similar to Theorem~\ref{thr:thr2.1}. Recall that $m = |I|$. The expected value of $X$ is
	 \begin{align*}
	    \mathbb{E}[X] &= m \cdot\mathbb{E}_j[X(j)]\\
	    &= \alpha m\cdot\mathbb{E}_j [(-1)^{m_j} w_{j}]\\
	    &= \alpha m\frac{t}{n}\sum_{j=1}^{n/t}\left(\operatorname{Pr}[m_j=0] - \operatorname{Pr}[m_j=1]\right) w_{j}\\
	    &= \alpha m\frac{t}{n}\left[\sum_{j:w_j=1}\frac{p(\sigma(x)^{(j)})}{\|A\|(t+1)} - \sum_{j:w_j=-1}\frac{p(\sigma(x)^{(j)})}{\|A\|(t+1)}\right].
	\end{align*}
	If $f(\sigma(x)^{(j)}) = w_j$, then $w_j = 1 \implies p(\sigma(x)^{(j)}) \geq \beta > 0$ and $w_j = -1 \implies p(\sigma(x)^{(j)}) \leq -\beta < 0 $. Therefore
	\begin{align*}
	    \mathbb{E}[X] &\geq \alpha m\frac{t}{n}\frac{1}{\|A\|(t+1)}\left[\sum_{j:w_j=1}\beta - \sum_{j:w_j=-1}-\beta\right] = \frac{\alpha \beta m}{\|A\|(t+1)}.
	\end{align*}
	If, on the other hand, $f(\sigma(x)^{(j)}) = -w_j$, then $w_j = 1 \implies p(\sigma(x)^{(j)}) \leq -\beta < 0$ and $w_j = -1 \implies p(\sigma(x)^{(j)}) \geq \beta > 0$. Therefore
	\begin{align*}
	    \mathbb{E}[X] &\leq \alpha m\frac{t}{n}\frac{1}{\|A\|(t+1)}\left[\sum_{j:w_j=1}-\beta - \sum_{j:w_j=-1}\beta\right] = -\frac{\alpha \beta m}{\|A\|(t+1)}.
	\end{align*}
    By using a Chernoff bound~\cite{dubhashi09} of the type $\operatorname{Pr}[X > \mathbb{E}[X] + u], \operatorname{Pr}[X < \mathbb{E}[X] - u] \leq e^{-2u^2/m}$ with $u > 0$ and setting $u = \pm \mathbb{E}[X] > 0$, we can make 
    \begin{align*}
        \operatorname{Pr}[X > 0 | B_f(x,\sigma) = \overline{w}], ~\operatorname{Pr}[X < 0 | B_f(x,\sigma) = w] \leq \epsilon
    \end{align*}
    by taking $m = O\big((\frac{t}{\alpha\beta})^2\log\frac{1}{\epsilon}\big)$, where we use that $\|A\| \leq \|A\|_{\infty\to 1} \leq 3$ according to Theorem \ref{thr:thr3.1} (note that $\frac{\|Ax\|_2}{\|x\|_2} \leq \frac{\|Ax\|_1}{\|x\|_\infty}$, and taking the supremum over all $x$ on both sides gives $\|A\| \leq \|A\|_{\infty\to 1}$). Therefore Alice and Bob can decide if $B_f(x,\sigma) = w$ or $B_f(x,\sigma) = \overline{w}$ with error probability $\epsilon$ and $O\big((\frac{t}{\alpha\beta})^2\log\frac{1}{\epsilon}\log{n}\big)$ qubits of communication.
    \end{proof}
    
    \section{Reductions from the Boolean Hidden Matching problem}
    \label{sec:sec4}
	
	As mentioned before, in~\cite{gavinsky2007exponential} it was proven that the Boolean Hidden Partition problem using PARITY on $2$ bits (aka the BHM problem) is hard to solve, i.e., $R^1(\operatorname{BHM}^\alpha_n) = \Omega(\sqrt{n/\alpha})$. With this result alone it is possible to prove that the $f$-Boolean Hidden Partition problem for almost any symmetric Boolean function with $\operatorname{sdeg}(f) \geq 2$ is at least as hard to solve. This can be achieved via a simple reduction from the BHM problem to the $f\operatorname{-BHP}^{\alpha}_{t,n}$ problem with symmetric functions, which we shall show in this section.
	
	For this section, in a slight abuse of notation we define $|x| = |\{i:x_i=-1\}|$ to be the ``Hamming weight'' of $x$. Let $s,t\in\mathbb{N}$, with $s\leq t$. Consider a symmetric Boolean function $f_s:\{-1,1\}^t \to \{-1,1\}$ such that (without loss of generality) $f_s(1^n) = 1$ and
	\begin{align}\label{eq:eq5.1}
		f_s(x) = \begin{cases}
			+1 \quad\text{if}~ 0 \leq |x| \leq \theta_1 ~\text{or}&\theta_{2i} < |x| \leq \theta_{2i+1}, i= 1,2\dots, \lfloor s/2\rfloor,\\
			-1 \quad\text{if}&\theta_{2j-1} < |x| \leq \theta_{2j}, j=1,2,\dots,\lfloor (s+1)/2\rfloor,
		\end{cases}
	\end{align} 
	where $\theta_k\in\mathbb{N}$ for $k=1,\dots,s+1$ and $0\leq \theta_1 < \dots < \theta_{s} < \theta_{s+1} = t$ and $\theta_{k+1} - \theta_k \geq 1$ for all $k=1,\dots,s$. The following result from \cite{aspnes1994expressive} tells us that $\operatorname{sdeg}(f_s) = s$.
	
	\begin{lemma}[{\cite[Lemma 2.6]{aspnes1994expressive}}]
	    If $f$ is a symmetric function, then $\operatorname{sdeg}(f)$ is equal to the number of times $f$ changes sign when expressed as a univariate function in $\sum_i x_i$.
	\end{lemma}
	
	In order to reduce $\operatorname{BHM}^\alpha_n$ to $f_s\operatorname{-BHP}^{\alpha}_{t,n}$ we first need to reduce the function PARITY to $f_s$, i.e., we want that $\forall x'\in\{-1,1\}^2$, $\exists x\in\{-1,1\}^t$ such that $f_s(x) = \operatorname{PARITY}(x')$. The key combinatorial step to achieve this is shown in the next lemma.
	
	\begin{lemma}\label{lem:lem5.1}
	    Let $f_s:\{-1,1\}^t \to \{-1,1\}$ be the symmetric Boolean function from Eq.~(\ref{eq:eq5.1}) with $s\geq 2$ such that either $2| t$ or $\theta_2 - \theta_1 < t-1$. Then there exists $a,b\in\mathbb{N}$ such that $\forall x'\in\{-1,1\}^2$, $\exists x\in\{-1,1\}^t$ such that $f_s(x) = \operatorname{PARITY}(x')$ and $|x| = a|x'| + b$.
	\end{lemma}
	\begin{proof}
	    The condition that $\forall x'\in\{-1,1\}^2$, $\exists x\in\{-1,1\}^t$ such that $f_s(x) = \text{PARITY}(x')$ and $|x| = a|x'| + b$ is equivalent to
	    \begin{align}
		\begin{cases}
			|x'| = 0 \implies f_s(b) = 1,\\
			|x'| = 1 \implies f_s(a+b) = -1,\\
			|x'| = 2 \implies f_s(2a+b) = 1.
		\end{cases}\label{eq:eq5.2}
	    \end{align}
	    We divide the proof into two cases: either there exists $k^\ast\in\{1,\dots,s-1\}$ such that $\theta_{k^\ast+1} - \theta_{k^\ast}$ is odd or there does not exist such a $k^\ast$. Suppose first that such $k^\ast$ exists. Without loss of generality we can assume that $f_s(x) = -1$ for $\theta_{k^\ast} < |x| \leq \theta_{k^\ast+1}$, otherwise we just flip the values of $f_s$. Then we set
    	\begin{align*}
		\begin{cases}
			a = (\theta_{k^\ast+1} - \theta_{k^\ast} + 1)/2,\\
			b = \theta_{k^\ast}.
		\end{cases}
	    \end{align*}
	    First, $a,b\in \mathbb{N}$. Second, $a+b = (\theta_{k^\ast+1} + \theta_{k^\ast} + 1)/2$, hence $\theta_{k^\ast} < a+b \leq \theta_{k^\ast+1}$, since $\theta_{k^\ast+1} - \theta_{k^\ast} \geq 1$. And third, $2a+b = \theta_{k^\ast+1} + 1 \leq \theta_{k^\ast+2}$. Therefore all conditions from Eqs.~(\ref{eq:eq5.2}) are satisfied.
	
	    Now suppose that for all $k=1,\dots,s-1$ we have $2|(\theta_{k+1} - \theta_k)$. Define the bit $\delta = 1$ if $\theta_1 \neq 0$ and $\delta = 0$ otherwise, and set
	    \begin{align*}
		\begin{cases}
			a = (\theta_{2} - \theta_{1} + 2)/2,\\
			b = \theta_{1} - \delta.
		\end{cases}
	    \end{align*}
	    First, $a,b\in \mathbb{N}$ (note that $\delta = 1 \implies \theta_1 > 0$). Second, $a+b = (\theta_{2} + \theta_{1} + 2 - 2\delta)/2$, hence $\theta_{1} < a+b \leq \theta_{2}$, since $\theta_{2} - \theta_{1} \geq 2$ by hypothesis. And third, $2a+b = \theta_{2} + 2 - \delta \leq t$ since $\theta_2 - \theta_1 \leq t-1$ and $\theta_2 < t$ (so that $\theta_2 = t-1 \implies \delta = 1$). Therefore all conditions from Eqs.~(\ref{eq:eq5.2}) are satisfied.
    \end{proof}
	
	If $2\nmid t$ and $\theta_2 - \theta_1 = t-1$, then our conditions give us
	\begin{align*}
		\begin{cases}
			b = 0,\\
			0 < a < t,\\
			2a = t,
		\end{cases}
	\end{align*}
	and we see that the condition $2a = t$ cannot be fulfilled by $a\in\mathbb{N}$. This case corresponds to the symmetric Boolean function Not All Equal (NAE), defined by $\operatorname{NAE}(x) = 1$ if $|x| \in \{0,t\}$ and $\operatorname{NAE}(x) = -1$ otherwise, with $t$ odd.
	
	Given the reduction above from PARITY to $f_s$, we can construct our reduction from $\operatorname{BHM}^\alpha_n$ to $f_s\operatorname{-BHP}^{\alpha}_{t,n}$. In the following theorem, we write $y^{(j;t)} = (y_{(j-1)t+1},y_{(j-1)t+2},\dots,y_{jt})\in\{-1,1\}^t$ to stress the size $t$ of the blocks from $y$ in order to better differentiate between strings in the reduction.
	
	\begin{theorem}\label{thr:thr4}
	    Let $f_s:\{-1,1\}^t \to \{-1,1\}$ be the symmetric Boolean function from Eq.~(\ref{eq:eq5.1}) with $s\geq 2$ such that either $2| t$ or $\theta_2 - \theta_1 < t-1$. Then $R^1(f_s\operatorname{-BHP}^{\alpha}_{t,n}) = \Omega(\sqrt{n/(\alpha t)})$.
	\end{theorem}
	\begin{proof}
	    Suppose towards a contradiction that $R^1(f_s\operatorname{-BHP}^{\alpha}_{t,n}) = o(\sqrt{n/(\alpha t)})$, i.e., there exists a protocol $\mathcal{P}$ that solves $f_s\operatorname{-BHP}^{\alpha}_{t,n}$ with $o(\sqrt{n/(\alpha t)})$ bits of communication. We are going to show that such protocol would allow Alice and Bob to solve the $\operatorname{BHM}^\alpha_n$ problem with $o(\sqrt{n/\alpha})$ bits of communication, a contradiction.
	    
	    Let $a,b\in\mathbb{N}$ be the numbers used in reducing PARITY to $f_s$ in Lemma~\ref{lem:lem5.1}. Alice increases her bit-string $x\in\{-1,1\}^n$ as follows: she makes $a$ copies of $x$, obtaining $x^{a}\in\{-1,1\}^{an}$, where $x^a = xx\cdots x$ represents $x$ repeated $a$ times. She then adds $bn/2$ times the bit $1$, obtaining $x^a1^{bn/2}$. Finally, she adds $(t-2a-b)n/2$ times the bit $-1$, to finally obtain $x_f = x^a1^{bn/2}\text{-}1^{(t-2a-b)n/2}$. Note that $x_f \in \{-1,1\}^{nt/2}$.
	    
	    Bob, on the other hand, increases his permutation $\sigma\in \mathbb{S}_n$ to a new permutation $\sigma_f\in \mathbb{S}_{nt/2}$. In order to describe how he does this, we shall ease the notation by referring to the $j$-th block $(\pi^{-1}((j-1)t+1),\dots,\pi^{-1}(jt))$ of a given permutation $\pi$ as $(B_{j,1},\dots,B_{j,t})$. With this notation, the $j$-th block $(B_{j,1},B_{j,2})$ of Bob's permutation $\sigma$ is mapped to the $j$-th block
	    \begin{align*}
	        \begin{multlined}[b][\textwidth]
	        \bigg(B_{j,1},~B_{j,2},~n+B_{j,1},~n+B_{j,2},~\dots,~(a-1)n+B_{j,1},~(a-1)n + B_{j,2}, \\
	        an + j,~an + j + \frac{n}{2},~\dots,~an + j + (t-2a-1)\frac{n}{2}\bigg)
	        \end{multlined}
	    \end{align*}
	    %
	    of the new permutation $\sigma_f$. Note that the new block has $t$ elements, as expected.
	    
	    Consider the block strings $\sigma_f(x_f)^{(j;t)}\in\{-1,1\}^t$ and $\sigma(x)^{(j;2)}\in\{-1,1\}^2$, with $j=1,\dots,n/2$. By construction we have that $|\sigma_f(x_f)^{(j;t)}| = a|\sigma(x)^{(j;2)}| + b$ and, according to Lemma~\ref{lem:lem5.1}, we get $f_s(\sigma_f(x_f)^{(j;t)}) = \operatorname{PARITY}(\sigma(x)^{(j;2)})$ for all $j=1,\dots,n/2$. Hence we see that every instance of the problem $\text{BHM}_n^{\alpha}:\{-1,1\}^n \to \{-1,1\}$ is mapped to an instance of the problem $f_s\text{-BHP}^{\alpha}_{t,nt/2}:\{-1,1\}^{nt/2} \to \{-1,1\}$. Therefore, the use of a protocol $\mathcal{P}$ that solves $f_s\text{-BHP}^{\alpha}_{t,tn/2}$ with $o(\sqrt{n/\alpha })$ bits of communication would lead to $R^1(\operatorname{BHM}_n^\alpha) = o(\sqrt{n/\alpha })$, which is impossible. Thus $R^1(f_s\operatorname{-BHP}^{\alpha}_{t,n}) = \Omega(\sqrt{n/(\alpha t)})$.
	\end{proof}
	
	\section{Preliminaries on Fourier Analysis}
    
    Since the next two sections rely heavily on discrete Fourier analysis~\cite{o2014analysis,de2008brief}, we shall provide some standard definitions and results from the area. Given functions $f,g\in\{-1,1\}^n\to\mathbb{R}$, we define their inner product by
    \begin{align*}
        \langle f,g\rangle = \operatorname*{\mathbb{E}}_{x\in\{-1,1\}^n}[f(x)g(x)] = \frac{1}{2^n}\sum_{x\in\{-1,1\}^n}f(x)g(x),
    \end{align*}
    and their $\ell_2$-norm by
    \begin{align*}
        \|f\|_2^2 = \langle f,f\rangle = \frac{1}{2^n}\sum_{x\in\{-1,1\}^n}f(x)^2.
    \end{align*}
    The Fourier transform of $f$ is a function $\widehat{f}:\{0,1\}^n\to\mathbb{R}$ defined by
    \begin{align*}
        \widehat{f}(S) = \langle f,\chi_S\rangle = \frac{1}{2^n}\sum_{x\in\{-1,1\}^n} f(x)\chi_S(x),
    \end{align*}
    where $\chi_S:\{-1,1\}^n\to\{-1,1\}$ is the character function $\chi_S(x) = \prod_{i\in S} x_i$ with $S\subseteq [n]$. The quantity $\widehat{f}(S)$ is the Fourier coefficient of $f$ corresponding to $S$. Thus we can write
    \begin{align*}
        f = \sum_{S\subseteq[n]}\widehat{f}(S)\chi_S.
    \end{align*}
    In our lower bound proofs we will need Parseval's identity and the important KKL inequality~\cite{kahn1988influence}.
    \begin{lemma}[Parseval's]\label{lem:lem4.1}
        For every function $f:\{-1,1\}^n\to\mathbb{R}$ we have $\|f\|_2^2 = \sum_{S\subseteq[n]}\widehat{f}(S)^2$.
    \end{lemma}
    \begin{lemma}[KKL]
        \label{lem:kkl}
        Let $f:\{-1,1\}^n\to\{-1,0,1\}$ and $A = \{x|f(x) \neq 0\}$. Then, for every $\delta \in[0,1]$,
        \begin{align*}
            \sum_{S\subseteq[n]} \delta^{|S|}\widehat{f}(S)^2 \leq \left(\frac{|A|}{2^n}\right)^{\frac{2}{1+\delta}}.
        \end{align*}
    \end{lemma}
    
    The above concepts can be generalised to matrix-valued functions, in particular the ones that map $x\in\{-1,1\}^n$ to a $m$-qubit density operator. The state space $\mathcal{A}$ of $m$ qubits is the complex Euclidean space $\mathbb{C}^{2^m}$. The set of all mixed quantum states in $\mathcal{A}$ is denoted by $\operatorname{D}(\mathcal{A})$. The \emph{Fourier transform} $\widehat{f}$ of a matrix-valued function $f:\{-1,1\}^n\to\operatorname{D}(\mathcal{A})$ is defined similarly as for scalar functions: is the function $\widehat{f}:\{-1,1\}^n\to\operatorname{D}(\mathcal{A})$ defined by
    \begin{align*}
         \widehat{f}(S) = \frac{1}{2^n}\sum_{x\in\{-1,1\}^n} f(x)\chi_S(x).
    \end{align*}
    In our quantum lower bound proof we will use the following special case of an extension of the hypercontractive inequality to matrix-valued functions \cite{ben2008hypercontractive}, where, for any $X\in\operatorname{D}(\mathcal{A})$ with singular values $\sigma_1,\dots,\sigma_d$, with $d=\operatorname{dim}(\mathcal{A})$, its \emph{trace norm} is defined as $\|X\|_{\operatorname{tr}} := \sum_{i=1}^d \sigma_i$.
    \begin{theorem}[{\cite[Lemma 6]{ben2008hypercontractive}}]
        \label{lem:hyper}
        For every $f:\{-1,1\}^n\to\operatorname{D}(\mathbb{C}^{2^m})$ and $\delta\in[0,1]$,
        \begin{align*}
            \sum_{S\subseteq[n]} \delta^{|S|}\|\widehat{f}(S)\|_{\operatorname{tr}}^2 \leq 2^{2\delta m}.
        \end{align*}
    \end{theorem}
    
    \section{Classical Lower Bound}
    \label{sec:sec5}
    
    Given a Boolean function $f:\{-1,1\}^t\to\{-1,1\}$, in this section we shall prove that the associated $f$-Boolean Hidden Partition problem is hard if $\operatorname{phdeg}(f) \geq 2$. The proof follows the general idea of~\cite{gavinsky2007exponential, verbin2011streaming}. Recall that the total variation distance\footnote{This distance is often defined with a factor of $1/2$; here we use the same normalisation as~\cite{gavinsky2007exponential}.} between two distributions $D$ and $D'$ is defined as $\|D - D'\|_{\operatorname{tvd}} := \sum_x |D(x) - D'(x)|$.
    
    In the classical and quantum proofs, given an expression $\mathcal{E}$, we denote by $\mathbf{1}[\mathcal{E}]$ the Iverson bracket, i.e., the indicator function $\mathbf{1}[\mathcal{E}] = 1$ if $\mathcal{E}$ is true and $0$ if not. 
    
    We briefly mention that the upper bound below on $\alpha$ comes from technical reasons that will arise during the proof. It is not tight, though, and could possibly be improved up to $\alpha \leq 1/2$.

\begin{theorem}\label{thr:thr5.1}
    If $\operatorname{phdeg}(f) = d \geq 2$ and $\alpha \leq \min\big(1/2, \frac{t}{2d}\hat{\|}f\hat{\|}_1^{-2/d}\big)$, where $\hat{\|}f\hat{\|}_1 := \sum_{T\subseteq[t]}|\widehat{f}(T)|$, then $R^1_\epsilon(f\operatorname{-BHP}^{\alpha}_
    {t,n}) = \Omega\big((\epsilon^4/(\alpha\hat{\|}f\hat{\|}_1^2))^{1/d}(n/t)^{1-1/d}\big)$.
\end{theorem}
\begin{proof}
    By the minimax principle, it suffices to analyse \emph{deterministic} protocols under some `hard' input distribution. For our input distribution, Alice and Bob receive $x\in\{-1,1\}^n$ and $\sigma\in\mathbb{S}_n$, respectively, uniformly at random, while Bob's input $w\in\{-1,1\}^{\alpha n/t}$ equals $B_f(x,\sigma)$ with probability $1/2$ and $\overline{B_f(x,\sigma)}$ with probability $1/2$. 
    	    
	Fix a small constant $\epsilon>0$ and let $c$ to be determined later. Consider any classical deterministic protocol that communicates at most $C :=  c - \log(1/\epsilon)$ bits. Such protocol partitions the set of all $2^n$ $x$'s into $2^C$ sets. These sets have size $2^{n-C}$ on average, and by a counting argument, with probability $1-\epsilon$, the set $A$ corresponding to Alice's message has size at least $\epsilon 2^{n-C} = 2^{n-c}$. 
	
	Given $A\subseteq\{-1,1\}^n$ such that $|A| \geq 2^{n-c}$, let $X$ be uniformly distributed over $A$ and let $Z = B_f(X,\sigma)$ given Bob's permutation $\sigma$. Bob, by looking at $w$, needs to decide whether $w = Z$ or $w = \overline{Z}$, i.e., he needs to discriminate between the following two induced distributions,
    \begin{align}
        \label{eq:pqdistr}
        p_\sigma(z) = \frac{|\{x\in A|B_f(x,\sigma) = z\}|}{|A|} ~\text{and}~ q_\sigma(z) = \frac{|\{x\in A| B_f(x,\sigma) = \overline{z}\}|}{|A|}.
    \end{align}
    He can only achieve this if the distributions have large total variation distance, since it is well known that the best success probability for distinguishing two distributions $q_1$ and $q_2$ with one sample is $1/2 + \|q_1 - q_2\|_{\operatorname{tvd}}/4$. We prove in Theorem~\ref{thr:raryhidden2} below that, if $\alpha\leq \min\big(1/2,\frac{t}{2d}\hat{\|}f\hat{\|}_1^{-2/d}\big)$ and $c \leq \gamma (\epsilon^4/(\alpha\hat{\|}f\hat{\|}_1^2))^{1/d}(n/t)^{1-1/d}$ for some universal constant $\gamma$, then the average advantage over all permutations $\sigma$ is at most $\epsilon^2/4$:
	\begin{align*}
		\operatorname*{\mathbb{E}}_{\sigma\sim\mathbb{S}_n}[\|p_\sigma - q_\sigma\|_{\text{tvd}}] \leq \epsilon^2.
	\end{align*}
	Therefore, by Markov's inequality, for at least a $(1-\epsilon)$-fraction of $\sigma$, the advantage in distinguishing between $p_\sigma$ and $q_\sigma$ is $\epsilon/4$ small. Hence, Bob's total advantage over randomly guessing the right distribution will be at most $\epsilon$ (for the event that $A$ is too small) plus $\epsilon$ (for the event that $\sigma$ is such that the distance between $p_\sigma$ and $q_\sigma$ is more than $\epsilon$) plus $\epsilon/4$ (for the advantage over random guessing when $\|p_\sigma - q_\sigma\|_{\text{tvd}} \leq \epsilon$), and so $c = \Omega\big( (\epsilon^4/(\alpha\hat{\|}f\hat{\|}_1^2))^{1/d}(n/t)^{1-1/d}\big)$.
\end{proof}

\begin{theorem}
	\label{thr:raryhidden2}
    Let $x\in\{-1,1\}^n$ be uniformly distributed over a set $A\subseteq \{-1,1\}^n$ of size $|A| \geq 2^{n-c}$ for some $c\geq 1$. Consider the distributions $p_\sigma$ and $q_\sigma$ from Eq.~\eqref{eq:pqdistr}. If $\alpha \leq \min\big(1/2,\frac{t}{2d}\hat{\|}f\hat{\|}_1^{-2/d}\big)$, then there is a universal constant $\gamma>0$ (independent of $n$, $t$, $d$ and $\alpha$), such that, for all $\epsilon > 0$, if $c \leq \gamma(\epsilon^4/(\alpha \hat{\|}f\hat{\|}_1^2))^{1/d} (n/t)^{1-1/d}$, then
    \begin{align*}
        \operatorname*{\mathbb{E}}_{\sigma\sim\mathbb{S}_n}[\|p_\sigma - q_\sigma\|_{\operatorname{tvd}}] \leq \epsilon^2.
    \end{align*}
\end{theorem}
\begin{proof}
    We start upper bounding the total variation distance by using Jensen's inequality,
    \begin{align*}
        \operatorname*{\mathbb{E}}_{\sigma\sim\mathbb{S}_n}[\|p_\sigma - q_\sigma\|_{\operatorname{tvd}}] \leq \sqrt{\operatorname*{\mathbb{E}}_{\sigma\sim\mathbb{S}_n}[\|p_\sigma - q_\sigma\|_{\operatorname{tvd}}^2]},
    \end{align*}
    and then by the Cauchy-Schwarz inequality and Parseval's identity (Lemma \ref{lem:lem4.1}) we finally get
    \begin{align*}
        \operatorname*{\mathbb{E}}_{\sigma\sim\mathbb{S}_n}[\|p_\sigma - q_\sigma\|_{\operatorname{tvd}}^2] \leq 2^{2\alpha n/t}\operatorname*{\mathbb{E}}_{\sigma\sim\mathbb{S}_n}[\|p_\sigma - q_\sigma\|_{2}^2] = 2^{2\alpha n/t}\operatorname*{\mathbb{E}}_{\sigma\sim\mathbb{S}_n}\left[\sum_{\substack{V\subseteq [\alpha n/t]\\ V\neq \emptyset}} \widehat{r_\sigma}(V)^2\right],
    \end{align*}
    where $r_\sigma(z) := p_\sigma(z) - q_\sigma(z)$, and we observed that $\widehat{p_{\sigma}}(\emptyset) = \widehat{q_{\sigma}}(\emptyset)$, as $p_\sigma$ and $q_\sigma$ are probability distributions. Let $g:\{-1,1\}^n\to\{0,1\}$ be the characteristic function of $A$, i.e., $f(x)=1$ iff $x\in A$. We can show that $p_\sigma$ and $q_\sigma$ are close for most permutations $\sigma$ by bounding the Fourier coefficients of $r_\sigma$, which are related to the Fourier coefficients of $g$ as follows.
    \begin{align*}
	    \widehat{r_\sigma}(V) &= \frac{1}{2^{\alpha n/t}} \sum_{z \in \{-1,1\}^{\alpha n/t}} \left(p_{\sigma}(z) - q_\sigma(z)\right) \chi_V(z)\\\displaybreak[0]
    	&= \frac{1}{|A| 2^{\alpha n/t}} \sum_{\substack{z \in \{-1,1\}^{\alpha n/t}\\x\in\{-1,1\}^n}} \mathbf{1}[x \in A] \left(\mathbf{1}[B_f(x,\sigma) = z] - \mathbf{1}[B_f(x,\sigma) = \overline{z}]\right) \chi_V(z)\\\displaybreak[2]
	    &= \frac{2}{|A| 2^{\alpha n/t}} \sum_{x \in \{-1,1\}^n} g(x) \chi_V(B_f(x,\sigma))
    \end{align*}
    for $|V|$ odd, otherwise $\widehat{r_\sigma}(V) = 0$, since $\chi_V(\overline{z}) = (-1)^{|V|}\chi_V(z)$. Using the Fourier expansion of $f$ and setting $|V| = k$, we have
	\begin{align*}
		\displaybreak[0]\chi_V(B_f(x,\sigma)) &= \prod_{j \in V} \left(\sum_{T \subseteq [t]} \widehat{f}(T) \chi_T(\sigma(x)^{(j)}) \right)\\\displaybreak[0]
		&=\prod_{j = 1}^k \left(\sum_{T_j \subseteq [t]} \widehat{f}(T_j) \chi_{T_j}(\sigma(x)^{(V_j)}) \right) \\\displaybreak[0]
		&=\sum_{T_1,\dots,T_k \subseteq [t]} \widehat{f}(T_1) \cdots \widehat{f}(T_k) \chi_{T_1}(\sigma(x)^{(V_1)}) \cdots \chi_{T_k}(\sigma(x)^{(V_k)})\\\displaybreak[0]
		&=\sum_{T_1,\dots,T_k \subseteq [t]} \widehat{f}(T_1) \cdots \widehat{f}(T_k) \chi_{V_1[T_1] \cup V_2[T_2] \cup \dots \cup V_k[T_k]}(\sigma(x))\\\displaybreak[0]
		&=\sum_{T_1,\dots,T_k \subseteq [t]} \widehat{f}(T_1) \cdots \widehat{f}(T_k) \chi_{V \bullet T}(\sigma(x)),
	\end{align*}
	where at the end we use the notation $V_i[T_i]$ to denote subset $T_i$ being positioned in block $V_i$, and then use the notation $V \bullet T$ to compactly represent $V_1[T_1] \cup V_2[T_2] \cup \dots \cup V_k[T_k]$. So we have
	\begin{align*}
		\widehat{r_\sigma}(V) &= \frac{2}{|A| 2^{\alpha n/t}} \sum_{\substack{x \in \{-1,1\}^n\\T_1,\dots,T_k \subseteq [t]}} g(x) \widehat{f}(T_1) \cdots \widehat{f}(T_k) \chi_{V \bullet T}(\sigma(x)) \\
		&= \frac{2^{n+1}}{|A| 2^{\alpha n/t}}\sum_{T_1,\dots,T_k \subseteq [t]} \widehat{f}(T_1) \cdots \widehat{f}(T_k) \widehat{g}(\sigma^{-1}(V\bullet T)),
	\end{align*}
	using that, for all $U \subseteq [n]$,
	\begin{align*}
	    \chi_U(\sigma(x)) = \prod_{i \in U} \sigma(x)_i = \prod_{i \in U} x_{\sigma^{-1}(i)} = \prod_{\sigma(j) \in U} x_j = \prod_{j \in \sigma^{-1}(U)} x_j.
	\end{align*}
	With all that,
	\begin{align}
	    \frac{1}{4}\operatorname*{\mathbb{E}}_{\sigma\sim\mathbb{S}_n}&[\|p_\sigma - q_\sigma\|_{\operatorname{tvd}}^2] \nonumber\\
	    &\leq \frac{2^{2n}}{|A|^2} \sum_{\substack{k=1 \\ k~\text{odd}}}^{\alpha n/t} \sum_{\substack{V \subseteq [\alpha n/t]\\|V| =k}} \operatorname*{\mathbb{E}}_{\sigma\sim\mathbb{S}_n} \left[ \left( \sum_{T_1,\dots,T_k \subseteq [t]} \widehat{f}(T_1) \cdots \widehat{f}(T_k)\cdot \widehat{g}(\sigma^{-1}(V \bullet T))\right)^2 \right]\nonumber \\\displaybreak[0]
	    &= \frac{2^{2n}}{|A|^2} \sum_{\substack{k=1 \\ k~\text{odd}}}^{\alpha n/t} \sum_{\substack{V \subseteq [\alpha n/t]\\|V| =k}} \sum_{\substack{T_1,\dots,T_k \subseteq [t] \\ U_1,\dots,U_k \subseteq[t]}}\operatorname*{\mathbb{E}}_{\sigma\sim\mathbb{S}_n}\left[\widehat{g}(\sigma^{-1}(V \bullet T))\widehat{g}(\sigma^{-1}(V \bullet U))\right]\prod_{j=1}^k \widehat{f}(T_j)\widehat{f}(U_j)\nonumber \\
	    &\leq \frac{2^{2n}}{|A|^2} \sum_{\substack{k=1 \\ k~\text{odd}}}^{\alpha n/t} \sum_{\substack{V \subseteq [\alpha n/t]\\|V| =k}}\left(\sum_{T_1,\dots,T_k \subseteq [t]}\sqrt{\operatorname*{\mathbb{E}}_{\sigma\sim\mathbb{S}_n}\big[\widehat{g}(\sigma^{-1}(V \bullet T))^2\big]}\prod_{j=1}^k |\widehat{f}(T_j)|\right)^2, \label{eq:eq4.2a}
	\end{align}
	where we used that $\mathbb{E}[XY] \leq \sqrt{\mathbb{E}[X^2]\mathbb{E}[Y^2]}$. Now we use the following combinatorial lemma.
	
	\begin{lemma}
	    For all $S \subseteq [n]$ and all $T_1,\dots,T_k\subseteq[t]$ such that $|S| = \sum_{j=1}^k |T_j|$,
	    \begin{align*}
	        \operatorname*{\mathbb{E}}_{\sigma\sim\mathbb{S}_n}\big[\widehat{g}(\sigma^{-1}(V \bullet T))^2\big] = \frac{1}{\binom{n}{|S|}}\sum_{\substack{S\subseteq[n] \\ |S| = \sum_{j=1}^k |T_j|}} \widehat{g}(S)^2.
	    \end{align*}
	\end{lemma}
	\begin{proof}
	By the definition of expected value,
	\begin{align*}
		\displaybreak[0]\operatorname*{\mathbb{E}}_{\sigma\sim\mathbb{S}_n} \left[\widehat{g}(\sigma^{-1}(V\bullet T))^2\right] &= \frac{1}{n!}\sum_{\sigma\in \mathbb{S}_n}\widehat{g}(\sigma^{-1}(V\bullet T))^2\\\displaybreak[0]
		&= \frac{1}{n!}\sum_{\sigma\in \mathbb{S}_n}\sum_{\substack{S\subseteq [n] \\ |S| = \sum_{j=1}^k |T_j|}} \mathbf{1}[\sigma^{-1}(V\bullet T) = S]\cdot  \widehat{g}(S)^2\\\displaybreak[0]
		&= \frac{1}{\binom{n}{|S|}}\sum_{\substack{S\subseteq [n] \\ |S| = \sum_{j=1}^k |T_j|}}  \widehat{g}(S)^2,
	\end{align*}
	since $\sum_{\sigma\in \mathbb{S}_n}\mathbf{1}[\sigma^{-1}(V\bullet T) = S] = |\{\sigma \in \mathbb{S}_n : \sigma^{-1}(V\bullet T) = S\}| = |S|!(n-|S|)!$.
	\end{proof}

	Using this Lemma in Eq.~(\ref{eq:eq4.2a}), we have
	\begin{align}
	    \frac{1}{4}\operatorname*{\mathbb{E}}_{\sigma\sim\mathbb{S}_n}[\|p_\sigma - q_\sigma\|_{\operatorname{tvd}}^2] &\leq \frac{2^{2n}}{|A|^2} \sum_{\substack{k=1 \\ k~\text{odd}}}^{\alpha n/t} \sum_{\substack{V \subseteq [\alpha n/t]\\|V| =k}} \left(\sum_{T_1,\dots,T_k \subseteq [t]}\sqrt{\sum_{\substack{S\subseteq[n]\\ |S|  = \sum_{i=1}^{k} |T_i|}}\frac{\widehat{g}(S)^2}{\binom{n}{|S|}}}\prod_{j=1}^k |\widehat{f}(T_j)|\right)^2 \nonumber\\\displaybreak[0]
		&\leq \frac{2^{2n}}{|A|^2} \sum_{\substack{k=1 \\ k~\text{odd}}}^{\alpha n/t} \sum_{\substack{V \subseteq [\alpha n/t]\\|V| =k}} \left(\sqrt{\sum_{\substack{S\subseteq[n]\\ kd \leq |S| \leq kt}}\frac{\widehat{g}(S)^2}{\binom{n}{|S|}}}\sum_{T_1,\dots,T_k \subseteq [t]}\prod_{j=1}^k |\widehat{f}(T_j)|\right)^2 \label{eq:eq4.21a} \\\displaybreak[0]
		&= \frac{2^{2n}}{|A|^2} \sum_{\substack{k=1 \\ k~\text{odd}}}^{\alpha n/t}\sum_{\substack{S \subseteq [n]\\kd\leq |S|\leq kt}}\frac{\binom{\alpha n/t}{k}}{\binom{n}{|S|}}\cdot \widehat{g}(S)^2 \left(\sum_{T_1,\dots,T_k \subseteq [t]} \prod_{j=1}^k|\widehat{f}(T_j)| \right)^2\nonumber\\
		&\leq \frac{2^{2n}}{|A|^2} \sum_{\substack{k=1 \\ k~\text{odd}}}^{\alpha n/t}\hat{\|}f\hat{\|}_1^{2k}\sum_{\substack{S \subseteq [n]\\kd\leq |S|\leq kt}}\frac{\binom{\alpha n/t}{k}}{\binom{n}{|S|}}\cdot\widehat{g}(S)^2,\label{eq:eq4.2}
	\end{align}
	where Eq.~(\ref{eq:eq4.21a}) comes from expanding the constraint $|S| = \sum_{j=1}^k |T_j|$ to the interval $kd \leq |S| \leq kt$, since $d \leq |T_j| \leq t$ for all $j\in[k]$ ($\widehat{f}(T_j) = 0$ if $|T_j|<\operatorname{phdeg}(f)=d$), so that the summation on $S\subseteq[n]$ can be pulled out of the summation on $T_1,\dots,T_k\subseteq[t]$, and in Eq.~(\ref{eq:eq4.2}) we denoted the sum of the Fourier masses of $f$ by $\hat{\|}f\hat{\|}_1 := \sum_{T\subseteq[t]} |\widehat{f}(T)|$. 

    We shall divide the above sum in two parts: one in the range $1\leq k < 2c$, and the other in the range $2c \leq k \leq \alpha n/t$.
    
    \paragraph{Sum I ($1\leq k < 2c$).} In order to upper bound the summation over $S\subseteq[n]$, consider $|S|=\ell$ for some $kd\leq \ell \leq kt$. Pick $\delta = \ell/2tc$ in Lemma~\ref{lem:kkl} (note that $\delta\in[0,1]$) to obtain (remember that $|A|\geq 2^{n-c}$)
	\begin{align*}
	    \frac{2^{2n}}{|A|^2}\sum_{\substack{S \subseteq [n]\\|S|=\ell}}\widehat{g}(S)^2 \leq \frac{2^{2n}}{|A|^2}\frac{1}{\delta^{\ell}}\sum_{S \subseteq [n]}\delta^{|S|}\widehat{g}(S)^2 \leq \frac{1}{\delta^\ell}\left(\frac{2^n}{|A|}\right)^{2\delta/(1+\delta)} \leq \frac{1}{\delta^\ell}\left(\frac{2^n}{|A|}\right)^{2\delta} \leq \left(\frac{2^{1/t} 2tc}{\ell}\right)^\ell.
	\end{align*}
	Therefore, and by using the upper bounds $\binom{an}{m} \leq a^m\binom{n}{m}$ for $a \leq 1$ to obtain Eq.~(\ref{eq:eq4.2aa}) below and then $\binom{n}{m}\binom{ln}{lm}^{-1} \leq (\frac{m}{n})^{(l-1)m}$ (see~\cite[Appendix A.5]{shi2012limits}) to obtain Eq.~(\ref{eq:eq4.2bb}), we can write
	\begin{align}
	    \frac{2^{2n}}{|A|^2} \sum_{\substack{k=1 \\ k~\text{odd}}}^{2c-1}\sum_{\substack{S \subseteq [n]\\kd\leq |S|\leq kt}}\hat{\|}f\hat{\|}_1^{2k}\frac{\binom{\alpha n/t}{k}}{\binom{n}{|S|}} \widehat{g}(S)^2 &\leq \sum_{\substack{k=1 \\ k~\text{odd}}}^{2c-1}\sum_{kd\leq |S|\leq kt} \hat{\|}f\hat{\|}_1^{2k}\frac{\binom{\alpha n/t}{k}}{\binom{n}{|S|}} \left(\frac{2^{1/t}2tc}{|S|}\right)^{|S|}\nonumber\\
	    &\leq \sum_{\substack{k=1 \\ k~\text{odd}}}^{2c-1}\sum_{kd\leq |S|\leq kt} \hat{\|}f\hat{\|}_1^{2k}\left(\frac{\alpha|S|}{kt}\right)^k\frac{\binom{kn/|S|}{k}}{\binom{n}{|S|}} \left(\frac{2^{1/t}2tc}{|S|}\right)^{|S|}\label{eq:eq4.2aa}\\
	    &\leq \sum_{\substack{k=1 \\ k~\text{odd}}}^{2c-1}\sum_{kd\leq |S|\leq kt} \hat{\|}f\hat{\|}_1^{2k}\left(\frac{\alpha|S|}{kt}\right)^k\left(\frac{|S|}{n}\right)^{|S|-k} \left(\frac{2^{1/t}2tc}{|S|}\right)^{|S|}\label{eq:eq4.2bb}\\
	    &= \sum_{\substack{k=1 \\ k~\text{odd}}}^{2c-1} \hat{\|}f\hat{\|}_1^{2k}\left(\frac{\alpha n}{kt}\right)^k \sum_{kd\leq |S|\leq kt}\left(\frac{2^{1/t}2c}{n/t}\right)^{|S|}\nonumber\\
	    &\leq \sum_{\substack{k=1 \\ k~\text{odd}}}^{2c-1} 2\hat{\|}f\hat{\|}_1^{2k}\left(\frac{\alpha n}{kt}\right)^k \left(\frac{2^{1/t}2c}{n/t}\right)^{kd},\nonumber
	\end{align}
	where we used a geometric series in the last step with $2^{1/t}2tc/n \leq 1/2$ for $n$ sufficiently large and the quantity $\binom{kn/|S|}{k}$ should be interpreted as $\frac{1}{k!}\prod_{j=0}^{k-1}\big(\frac{kn}{|S|}-j\big)$. Finally, and by using that $c \leq \gamma(\epsilon^4/(\alpha \hat{\|}f\hat{\|}_1^{2}))^{1/d} (n/t)^{1-1/d}$, we obtain
	\begin{align*}
        \frac{2^{2n}}{|A|^2} \sum_{\substack{k=1 \\ k~\text{odd}}}^{2c-1}\sum_{\substack{S \subseteq [n]\\kd\leq |S|\leq kt}}\hat{\|}f\hat{\|}_1^{2k}\frac{\binom{\alpha n/t}{k}}{\binom{n}{|S|}} \widehat{g}(S)^2 &\leq \sum_{\substack{k=1 \\ k~\text{odd}}}^{2c-1} 2\hat{\|}f\hat{\|}_1^{2k}\left(\frac{\alpha n}{kt}\right)^k \left(\frac{2^{1/t}2\gamma\epsilon^{4/d} (n/t)^{1-1/d}}{\alpha^{1/d}\hat{\|}f\hat{\|}_1^{2/d} n/t}\right)^{kd}\\
        &= \sum_{\substack{k=1 \\ k~\text{odd}}}^{2c-1}2\left(\frac{2^{1/t}2\gamma\epsilon^{4/d}}{k^{1/d}}\right)^{kd} \leq \frac{\epsilon^4}{8},
	\end{align*}
	where we took $\gamma$ sufficiently small in the last step.
	
	\paragraph{Sum II ($2c \leq k \leq \alpha n/t$).} Again using the approximation
	$$
	    \frac{\binom{\alpha n/t}{k}}{\binom{n}{kd}} \leq \left(\frac{2\alpha d}{t}\right)^{kd} \frac{\binom{\alpha n/t}{k}}{\binom{2\alpha nd/t}{kd}} \leq \left(\frac{2\alpha d}{t}\right)^{kd} \frac{\binom{2\alpha n/t}{k}}{\binom{2\alpha nd/t}{kd}}
	$$ 
	($\alpha \leq 1/2$, so $\frac{2\alpha d}{t} \leq 1$), note that the function $h(k) := \big(\frac{2\alpha d}{t}\hat{\|}f\hat{\|}_1^{2/d}\big)^{kd} \binom{2\alpha n/t}{k}/\binom{2\alpha nd/t}{kd}$ is non-increasing in the range $1\leq k \leq \alpha n/t$, since $\frac{2\alpha d}{t}\hat{\|}f\hat{\|}_1^{2/d} \leq 1$ by assumption and
	\begin{align*}
	    \frac{h(k-1)}{h(k)} \geq \frac{\binom{2\alpha n/t}{k-1}}{\binom{2\alpha nd/t}{kd-d}}\frac{\binom{2\alpha nd/t}{kd}}{\binom{2\alpha n/t}{k}} = \prod_{j=1}^{d-1}\frac{2\alpha nd/t-kd+j}{kd-j} \geq 1,
	\end{align*}
	since $k \leq \alpha n/t$ implies that $2\alpha nd/t-kd \geq kd$. Hence, by using Parseval's identity $\sum_{S\subseteq[n]} |\widehat{g}(S)|^2 = |A|/2^n$, the inequality $\binom{n}{m}\binom{ln}{lm}^{-1} \leq (\frac{m}{n})^{(l-1)m}$ and $c \leq \gamma(\epsilon^4/(\alpha \hat{\|}f\hat{\|}_1^{2}))^{1/d} (n/t)^{1-1/d}$, we have
	\begin{align*}
	    \frac{2^{2n}}{|A|^2} \sum_{\substack{k=2c \\ k~\text{odd}}}^{\alpha n/t}\sum_{\substack{S \subseteq [n]\\kd\leq |S|\leq kt}}\hat{\|}f\hat{\|}_1^{2k}\frac{\binom{\alpha n/t}{k}}{\binom{n}{|S|}} \widehat{g}(S)^2 &\leq h(2c) \frac{2^{2n}}{|A|^2} \sum_{\substack{k=2c \\ k~\text{odd}}}^{\alpha n/t}\sum_{\substack{S \subseteq [n]\\kd\leq |S|\leq kt}} \widehat{g}(S)^2\\
	    &\leq \frac{\alpha n}{t}2^c h(2c)\\
	    &\leq \frac{\alpha n}{t}2^c \left(\frac{2\alpha d}{t} \hat{\|}f\hat{\|}_1^{2/d}\right)^{2cd}\left(\frac{ct}{\alpha n}\right)^{(d-1)2c}\\
	    &\leq \frac{\alpha n}{t}2^c \left(\frac{2\alpha d}{t} \hat{\|}f\hat{\|}_1^{2/d}\right)^{2cd}\left(\frac{\gamma\epsilon^{4/d}(n/t)^{1-1/d}}{\alpha^{1/d}\hat{\|}f\hat{\|}_1^{2/d} \alpha n/t}\right)^{(d-1)2c}\\
	    &= \frac{\alpha n}{t}2^c \left(\frac{ (2d)^d}{t^d}(\alpha\hat{\|}f\hat{\|}_1^2)^{1/d}\right)^{2c}\left(\frac{\gamma\epsilon^{4/d}}{(n/t)^{1/d}}\right)^{(d-1)2c}\\
	    &\leq \frac{\epsilon^4}{8},
	\end{align*}
	where in the last step we used that $c\geq 1 \implies (1-1/d)2c \geq 1$ (so $n$ is in the denominator and $\epsilon^{(1-1/d)2c} \leq \epsilon$) and picked $\gamma$ sufficiently small.
	
	Summing both results, if $c \leq \gamma(\epsilon^4/(\alpha\hat{\|}f\hat{\|}_1^2))^{1/d} (n/t)^{1-1/d}$, then $\operatorname*{\mathbb{E}}_{\sigma\sim\mathbb{S}_n}[\|p_\sigma - q_\sigma\|^2_{\operatorname{tvd}}] \leq \epsilon^4$. By Jensen's inequality, we finally get $\operatorname*{\mathbb{E}}_{\sigma\sim\mathbb{S}_n}[\|p_\sigma - q_\sigma\|_{\operatorname{tvd}}] \leq \sqrt{\operatorname*{\mathbb{E}}_{\sigma\sim\mathbb{S}_n}[\|p_\sigma - q_\sigma\|^2_{\operatorname{tvd}}]} \leq \epsilon^2$.
    \end{proof}
	
	\section{Quantum Lower Bound}\label{sec:sec6.1}
	
	While in the previous section we proved a classical lower bound $\Omega(n^{1-1/d})$ for $f\operatorname{-BHP}^{\alpha}_{t,n}$ when $\operatorname{phdeg}(f) = d \geq 2$, in this section we shall prove its quantum analogue, Theorem~\ref{thr:thr6.1}. The first half of the proof follows the same line as the quantum lower bound for the $\operatorname{BHH}_{t,n}^\alpha$ problem from~\cite{shi2012limits}, but in the second half we used a similar approach from~\cite{gavinsky2007exponential,verbin2011streaming} of splitting the sum bounding the (average) bias in two parts. Differently from the classical lower bound proof, the Fourier analysis is performed directly on the encoding messages and not on the pre-images of a fixed encoding message, since there is no clear quantum analogue of conditioning on a message. Moreover, the matrix-valued hypercontractive inequality is now used in order to bound the trace norms $\|\widehat{\rho}\|_{\text{tr}}$.
	\begin{lemma}[\cite{helstrom1976quantum}]
	    \label{lem:lem3.5.c3}
	    Let $\rho_0,\rho_1$ be two quantum states which appear with probability $p$ and $1-p$, respectively. The optimal success probability of predicting which state it is by a POVM is
	    \begin{align*}
	        \frac{1}{2} + \frac{1}{2}\|p\rho_0 - (1-p)\rho_1\|_{\operatorname{tr}}.
	    \end{align*}
	\end{lemma}
\begin{theorem}\label{thr:thr6.1}
    If $\operatorname{phdeg}(f) = d \geq 3$ and $\alpha \leq \min(1/2,\frac{t}{2d}\hat{\|}f\hat{\|}_1^{-2/d})$, where $\hat{\|}f\hat{\|}_1 := \sum_{T\subseteq[t]}|\widehat{f}(T)|$, then $Q^1_\epsilon(f\operatorname{-BHP}^{\alpha}_{t,n}) = \Omega\big((\epsilon^2/(\alpha \hat{\|}f\hat{\|}_1))^{2/d} (n/t)^{1-2/d}\big)$.
\end{theorem}
\begin{proof}
    Consider an arbitrary $m$-qubit communication protocol, which can be viewed as Alice encoding her input $x\in \{-1,1\}^n$ into a quantum state and sending it to Bob so that Bob can distinguish if his input $w$ equals $B_f(x,\sigma)$ or $\overline{B_f(x,\sigma)}$.
    Let $\rho:\{-1,1\}^n\to \operatorname{D}(\mathbb{C}^{2^m})$ be Alice's encoding function. For our `hard' distribution, Alice and Bob receive $x\in\{-1,1\}^n$ and $\sigma\in\mathbb{S}_n$, respectively, uniformly at random, while Bob's input $w\in\{-1,1\}^{\alpha n/t}$ equals $B_f(x,\sigma)$ with probability $1/2$ or $\overline{B_f(x,\sigma)}$ with probability $1/2$. Let $p_x := 1/2^{n}$, $p_\sigma := 1/n!$ and $p_b := 1/2$, then our hard distribution $\mathcal{P}$ is
	\begin{align}
	    \label{eq:eq3.5.c3}
        \operatorname{Pr}[x,b,\sigma,w] = p_xp_\sigma p_b \mathbf{1}[B_f(x,\sigma) \circ b^{\alpha n/t} = w].
	\end{align}
	    
	Conditioning on Bob's input $(\sigma, w)$, from his perspective, Alice sends the message $\rho_x$ with probability $\operatorname{Pr}[x|\sigma,w]$. Therefore, conditioned on the value of $b$, Bob receives one of the following two quantum states $\rho_{+1}^{\sigma,w}$ and $\rho_{-1}^{\sigma,w}$, each appearing with probability $\operatorname{Pr}[b=+1|\sigma,w]$ and $\operatorname{Pr}[b=-1|\sigma,w]$, respectively,
	\begin{align}
	\label{eq:rhoyesrhono}
	\begin{aligned}
		\rho_{+1}^{\sigma,w} &= \sum_{x\in\{-1,1\}^n} \operatorname{Pr}[x|b=+1,\sigma,w]\cdot\rho_x = \frac{1}{\operatorname{Pr}[b=+1,\sigma,w]} \sum_{x\in\{-1,1\}^n} \operatorname{Pr}[x,b=+1,\sigma,w] \cdot \rho_x,\\ 
	    \rho_{-1}^{\sigma,w} &= \sum_{x\in\{-1,1\}^n} \operatorname{Pr}[x|b=-1,\sigma,w]\cdot\rho_x = \frac{1}{\operatorname{Pr}[b=-1,\sigma,w]} \sum_{x\in\{-1,1\}^n} \operatorname{Pr}[x,b=-1,\sigma,w] \cdot \rho_x.
	\end{aligned}
	\end{align}
	His best strategy to determine the value of $b$ conditioning on his input $(\sigma,w)$ is no more than the chance to distinguish between these two quantum states $\rho_{+1}^{\sigma,w}$ and $\rho_{-1}^{\sigma,w}$.
	
	Now let $\varepsilon_{bias}$ be the bias of the protocol that distinguishes between $\rho_{+1}^{\sigma,w}$ and $\rho_{-1}^{\sigma,w}$. According to Lemma~\ref{lem:lem3.5.c3}, the bias $\varepsilon_{bias}$ of any quantum protocol for a fixed $\sigma$ and $w$ can be upper bounded as
	\begin{align*}
		\varepsilon_{bias} \leq \big\|{\operatorname{Pr}}[b=+1|\sigma,w]\cdot\rho_{+1}^{\sigma,w} - \operatorname{Pr}[b=-1|\sigma,w]\cdot\rho_{-1}^{\sigma,w}\big\|_{\operatorname{tr}}.
	\end{align*}
	We prove in Theorem~\ref{thr:raryhidden} below that, if $\alpha \leq \min(1/2,\frac{t}{2d}\hat{\|}f\hat{\|}_1^{-2/d})$ and $m \leq \gamma(\epsilon^2/(\alpha \hat{\|}f\hat{\|}_1))^{2/d} (n/t)^{1-2/d}$ for a universal constant $\gamma$, then the average bias over $\sigma$ and $w$ is at most $\epsilon^2$, i.e.,
	\begin{align*}
		\operatorname*{\mathbb{E}}_{(\sigma,w)\sim\mathcal{P}_{\sigma,w}}[\varepsilon_{bias}] \leq \epsilon^2,
	\end{align*}
	where $\mathcal{P}_{\sigma,w}$ is the marginal distribution of $\mathcal{P}$. Therefore, by Markov's inequality, for at least a $(1-\epsilon)$-fraction of $\sigma$ and $w$, the bias in distinguishing between $\rho_{+1}^{\sigma,w}$ and $\rho_{-1}^{\sigma,w}$ is $\epsilon$ small. Therefore, Bob's advantage over randomly guessing the right distribution will be at most $\epsilon$ (for the event that $\sigma$ and $w$ are such that the distance between $\rho_{+1}^{\sigma,w}$ and $\rho_{-1}^{\sigma,w}$ is more than $\epsilon$) plus $\epsilon/2$ (for the advantage over random guessing when $\varepsilon_{bias} \leq \epsilon$), and so $m = \Omega\big((\epsilon^2/(\alpha \hat{\|}f\hat{\|}_1))^{2/d} (n/t)^{1-2/d}\big)$.
\end{proof}
\begin{theorem}
	\label{thr:raryhidden}
    For $x\in\{-1,1\}^n$, $\sigma\in\mathbb{S}_n$, $w\in\{-1,1\}^{\alpha n/t}$ and $b\in\{-1,1\}$, consider the probability distribution $\mathcal{P}$ defined in Eq.~(\ref{eq:eq3.5.c3}). Given an encoding function $\rho:\{-1,1\}^n\to \operatorname{D}(\mathbb{C}^{2^m})$, consider the quantum states $\rho_{+1}^{\sigma,w}$ and $\rho_{-1}^{\sigma,w}$ from Eq.~(\ref{eq:rhoyesrhono}). If $\alpha \leq \min(1/2,\frac{t}{2d}\hat{\|}f\hat{\|}_1^{-2/d})$, then there is a universal constant $\gamma>0$ (independent of $n$, $t$, $d$ and $\alpha$), such that, for all $\epsilon > 0$, if $m \leq \gamma(\epsilon^2/(\alpha \hat{\|}f\hat{\|}_1))^{2/d} (n/t)^{1-2/d}$, then
    \begin{align*}
        \operatorname*{\mathbb{E}}_{(\sigma,w)\sim\mathcal{P}_{\sigma,w}}\Big[\big\|\operatorname{Pr}[b=+1|\sigma,w]\cdot\rho_{+1}^{\sigma,w} - \operatorname{Pr}[b=-1|\sigma,w]\cdot\rho_{-1}^{\sigma,w}\big\|_{\operatorname{tr}}\Big] \leq \epsilon^2.
    \end{align*}
\end{theorem}
\begin{proof}
    For simplicity we shall write
	\begin{align*}
	    \varepsilon_{bias} := \operatorname*{\mathbb{E}}_{(\sigma,w)\sim\mathcal{P}_{\sigma,w}}\Big[\big\|\operatorname{Pr}[b=+1|\sigma,w]\cdot\rho_{+1}^{\sigma,w} - \operatorname{Pr}[b=-1|\sigma,w]\cdot\rho_{-1}^{\sigma,w}\big\|_{\operatorname{tr}}\Big].
	\end{align*}
    By the definition of $\varepsilon_{bias}$, we have that
    \begin{align*}
        \varepsilon_{bias} &= \sum_{\sigma\in\mathbb{S}_n}\sum_{w\in\{-1,1\}^{\alpha n/t}}\operatorname{Pr}[\sigma,w]\cdot \big\|\operatorname{Pr}[b=+1|\sigma,w]\cdot\rho_{+1}^{\sigma,w} - \operatorname{Pr}[b=-1|\sigma,w]\cdot\rho_{-1}^{\sigma,w}\big\|_{\operatorname{tr}}\\
        &= \sum_{\sigma\in\mathbb{S}_n}\sum_{w\in\{-1,1\}^{\alpha n/t}}\Big\|\sum_{x\in\{-1,1\}^n}\big(\operatorname{Pr}[x,b=+1,\sigma,w]\cdot\rho_x - \operatorname{Pr}[x,b=-1,\sigma,w]\cdot\rho_x\big)\Big\|_{\operatorname{tr}}\\
        &= \sum_{\sigma\in\mathbb{S}_n}\sum_{w\in\{-1,1\}^{\alpha n/t}}\Big\|\sum_{x\in\{-1,1\}^n}\frac{1}{2}p_xp_\sigma \left(\mathbf{1}\big[B_f(x,\sigma) = w\big] - \mathbf{1}\big[B_f(x,\sigma) = \overline{w}\big]\right)\rho_x\Big\|_{\text{tr}}\\
        &= \sum_{\sigma\in\mathbb{S}_n}\sum_{w\in\{-1,1\}^{\alpha n/t}}\Big\|\sum_{x\in\{-1,1\}^n}\frac{1}{2}p_xp_\sigma \left(\mathbf{1}\big[B_f(x,\sigma) = w\big] - \mathbf{1}\big[B_f(x,\sigma) = \overline{w}\big]\right)\sum_{S\subseteq[n]}\widehat{\rho}(S)\chi_S(x)\Big\|_{\text{tr}}\\
        &= \sum_{\sigma\in\mathbb{S}_n}\sum_{w\in\{-1,1\}^{\alpha n/t}} \Big\|\sum_{S\subseteq[n]} u(\sigma,w,S) \widehat{\rho}(S)\Big\|_{\text{tr}}\\
        &\leq \sum_{S\subseteq[n]}\sum_{\sigma\in\mathbb{S}_n}\sum_{w\in\{-1,1\}^{\alpha n/t}}|u(\sigma,w,S)|\|\widehat{\rho}(S)\|_{\text{tr}},
    \end{align*}
    where we used the Fourier expansion of $\rho_x$ and defined
	\begin{align*}
		u(\sigma,w,S) := \frac{1}{2}\sum_{x\in\{-1,1\}^n}p_xp_\sigma \chi_S(x) \left(\mathbf{1}\big[B_f(x,\sigma) = w\big] - \mathbf{1}\big[B_f(x,\sigma) = \overline{w}\big]\right).
	\end{align*}
    
    We now need to analyse $|u(\sigma,w,S)|$ for different $\sigma, w, S$. The way it is done is by `breaking' $\sigma(S)$ into blocks. Consider the set $[i|j] := \{i+1,i+2\dots,j\}$. Given $V \subseteq[n]$, we note that the set $V\cap [(j-1)t|jt]$ contains the elements of $V$ that are in the interval $[(j-1)t+1,jt]$.  From this we define $U_j\subseteq[t]$, for $j\in[n/t]$, as the set with elements from $V\cap [(j-1)t|jt]$ all shifted by $-(t-1)j$, so they fall in the interval $[1,t]$. It is clear that $V = \bigcup_{j=1}^{n/t} V\cap [(j-1)t|jt]$, which we shall write as $V = [n/t]\bullet \mathbb{U}_V$, where $\mathbb{U}_V = (U_1,\dots,U_{n/t}) \in [t]^{n/t}$. The sequence of sets $\mathbb{U}_V$ is giving the decomposition of $V$ into $n/t$ blocks of length $t$. From it we can define $\mathbb{U}^\ast_V$ as the sequence of entries from $\mathbb{U}_V$ that are nonempty. With these in mind, the quantity $|u(\sigma,w,S)|$ is given by the following lemma, which is proven at the end of the section.
    \begin{lemma}\label{lem:lem6.1}
        Given $S\subseteq[n]$ and $\sigma\in\mathbb{S}_n$, take $\sigma(S) = [n/t]\bullet \mathbb{U}_{\sigma(S)}$ with $\mathbb{U}_{\sigma(S)} = (U_1,\dots,U_{n/t}) \in [t]^{n/t}$. Define $\Delta := \{V\subseteq[\alpha n] : |\mathbb{U}_{V}^\ast| ~\text{odd and}~ |U| \geq d, ~\forall U \in \mathbb{U}_{V}^\ast\}$, where $\operatorname{phdeg}(f) = d$. Then
        \begin{align*}
	        |u(\sigma, w, S)| = \begin{cases}
		        \frac{p_\sigma}{2^{\alpha n/t}} \prod_{U\in \mathbb{U}^\ast_{\sigma(S)}}|\widehat{f}(U)| &~\text{if}~ \sigma(S)\in \Delta,\\
		        0 &~\text{if}~ \sigma(S)\notin \Delta.
	        \end{cases}
        \end{align*}
    \end{lemma}
    From Lemma~\ref{lem:lem6.1} we see that $u(\sigma,w,S)$ is only nonzero when $\sigma(S)\in \Delta$, so the expression for the bias becomes
    \begin{align}
        \varepsilon_{bias} &\leq \sum_{S\subseteq[n]}\sum_{\substack{\sigma\in\mathbb{S}_n \\ \sigma(S) \in \Delta}}p_\sigma \sum_{w\in\{-1,1\}^{\alpha n/t}}\frac{1}{2^{\alpha n/t}}\|\widehat{\rho}(S)\|_{\text{tr}}\prod_{U\in \mathbb{U}^\ast_{\sigma(S)}}|\widehat{f}(U)|\nonumber\\
        &= \sum_{S\subseteq[n]}\|\widehat{\rho}(S)\|_{\text{tr}}\sum_{\substack{\sigma\in\mathbb{S}_n \\ \sigma(S) \in \Delta}}p_\sigma \prod_{U\in \mathbb{U}^\ast_{\sigma(S)}}|\widehat{f}(U)|.\label{eq:eq6.01}
    \end{align}
    One of the requirements for $\sigma(S)\in \Delta$ is that the block decomposition of $\sigma(S)$ must have an odd number of nonempty blocks. Given $S$ and $\sigma$, the number of nonempty blocks of $\sigma(S)$ (which we will denote by the index $k$ below) is lower-bounded by $\lceil |S|/t\rceil \geq 1$ and upper-bounded by $\min(\lfloor|S|/d\rfloor, \alpha n/t) \leq \alpha n/t$ (since $\sigma(S)\subseteq[\alpha n]$). With these points in mind, and decomposing $\sigma(S)$ as $V\bullet \mathbb{U}$ for some $V\subseteq[\alpha n/t]$ with $|V|=k$ odd (number of blocks) and $\mathbb{U}=(U_1,\dots,U_k)$, we have
    \begin{align}
        \sum_{\substack{\sigma\in\mathbb{S}_n \\ \sigma(S) \in \Delta}}p_\sigma &\prod_{U\in \mathbb{U}^\ast_{\sigma(S)}}|\widehat{f}(U)| \nonumber\\
        &= \frac{1}{n!}\sum_{\sigma\in\mathbb{S}_n} [\sigma(S)\in \Delta]\prod_{U\in \mathbb{U}^\ast_{\sigma(S)}}|\widehat{f}(U)| \nonumber\\\displaybreak[0]
        &\leq \sum_{\substack{k = 1 \\ k~\text{odd}}}^{\alpha n/t}\frac{1}{n!}\sum_{\substack{U_1,\dots,U_k\subseteq[t]\\ d\leq |U_j| \leq t}} \mathbf{1}\left[\sum_{j=1}^k |U_j| = |S|\right]\sum_{\substack{V\subseteq[\alpha n/t] \\ |V|=k}}|\{\sigma\in \mathbb{S}_n : \sigma^{-1}(V\bullet \mathbb{U}) = S\}| \prod_{j=1}^k|\widehat{f}(U_j)|\nonumber\\\displaybreak[0]
        &= \sum_{\substack{k = 1 \\ k~\text{odd}}}^{\alpha n/t}\frac{\binom{\alpha n/t}{k}}{\binom{n}{|S|}}\sum_{\substack{U_1,\dots,U_k\subseteq[t]\\ d\leq |U_j| \leq t}} \mathbf{1}\left[\sum_{j=1}^k |U_j| = |S|\right]\prod_{j=1}^k|\widehat{f}(U_j)|, \label{eq:eq6.02}
    \end{align}
    since $|\{\sigma \in \mathbb{S}_n : \sigma^{-1}(V\bullet \mathbb{U}) = S\}| = |S|!(n-|S|)!$ for any $V$. By plugging Eq.~(\ref{eq:eq6.02}) into Eq.~(\ref{eq:eq6.01}),
    \begin{align*}
        \displaybreak[0]\varepsilon_{bias} &\leq \sum_{S\subseteq[n]}\sum_{\substack{k = 1 \\ k~\text{odd}}}^{\alpha n/t} \frac{\binom{\alpha n/t}{k}}{\binom{n}{|S|}} \|\widehat{\rho}(S)\|_{\text{tr}}\sum_{\substack{U_1,\dots,U_k\subseteq[t]\\ d\leq |U_j| \leq t}} \mathbf{1}\left[\sum_{j=1}^k |U_j| = |S|\right]\prod_{j=1}^k|\widehat{f}(U_j)| \\\displaybreak[0]
        &= \sum_{\substack{k = 1 \\ k~ \text{odd}}}^{\alpha n/t}\sum_{\substack{U_1,\dots,U_k\subseteq[t]\\ d\leq |U_j| \leq t}} \frac{\prod_{j=1}^k|\widehat{f}(U_j)|}{\binom{\alpha n/t}{k}} \sum_{\substack{S\subseteq[n] \\ |S| = \sum_{j=1}^k |U_j|}} \frac{\binom{\alpha n/t}{k}^2}{\binom{n}{|S|}} \|\widehat{\rho}(S)\|_{\text{tr}}.
    \end{align*}
    By Cauchy-Schwarz, 
    \begin{align*}
        \displaybreak[0]\sum_{\substack{S\subseteq[n] \\ |S| = \sum_{j=1}^k |U_j|}} \frac{\binom{\alpha n/t}{k}^2}{\binom{n}{|S|}} \|\widehat{\rho}(S)\|_{\text{tr}} &\leq \sqrt{\sum_{\substack{S\subseteq[n] \\ |S| = \sum_{j=1}^k |U_j|}} \frac{\binom{\alpha n/t}{k}^2}{\binom{n}{|S|}}}\sqrt{\sum_{\substack{S\subseteq[n] \\ |S| = \sum_{j=1}^k |U_j|}} \frac{\binom{\alpha n/t}{k}^2}{\binom{n}{|S|}} \|\widehat{\rho}(S)\|_{\text{tr}}^2}\\\displaybreak[0]
        &= \binom{\alpha n/t}{k}\sqrt{\sum_{\substack{S\subseteq[n] \\ |S| = \sum_{j=1}^k |U_j|}} \frac{\binom{\alpha n/t}{k}^2}{\binom{n}{|S|}} \|\widehat{\rho}(S)\|_{\text{tr}}^2}.\displaybreak[0]
    \end{align*}
    Therefore, by expanding the constraint $|S| = \sum_{j=1}^k |U_j|$ to the interval $kd \leq |S| \leq kt$, since $d \leq |U_j| \leq t$ for all $j\in[k]$ ($\widehat{f}(U_j) = 0$ if $|U_j|<\operatorname{phdeg}(f)=d$), and by denoting the sum of the Fourier masses of $f$ by $\hat{\|}f\hat{\|}_1 := \sum_{U\subseteq[t]} |\widehat{f}(U)|$, we get
    \begin{align*}
        \varepsilon_{bias} \leq \sum_{\substack{k = 1 \\ k~ \text{odd}}}^{\alpha n/t}\hat{\|}f\hat{\|}_1^k \sqrt{\sum_{\substack{S\subseteq[n] \\ kd \leq |S| \leq kt}} \frac{\binom{\alpha n/t}{k}^2}{\binom{n}{|S|}} \|\widehat{\rho}(S)\|_{\text{tr}}^2}.
    \end{align*}
    
    Similarly to the classical proof, we shall divide the above sum in two parts: one in the range $1\leq k < 4m$, and the other in the range $4m \leq k \leq \alpha n/t$.
    
    \paragraph{Sum I ($1\leq k < 4m)$.} In order to upper bound the summation over $S\subseteq[n]$, consider $|S|=\ell$ for some $kd\leq \ell \leq kt$. We pick $\delta = \ell/4tm$ in Lemma~\ref{lem:hyper} (note that $\delta\in[0,1]$) to obtain
    \begin{align*}
        \sum_{\substack{S\subseteq[n] \\ |S| = \ell}} \|\widehat{\rho}(S)\|_{\operatorname{tr}}^2 \leq \frac{1}{\delta^{\ell}}\sum_{S\subseteq[n]}\delta^{|S|} \|\widehat{f}(S)\|_{\operatorname{tr}}^2 \leq \frac{1}{\delta^{\ell}}2^{2\delta m} = \left(\frac{2^{1/(2t)}4tm}{\ell}\right)^{\ell}.
    \end{align*}
    Therefore, and by using the upper bounds $\binom{an}{m} \leq a^m\binom{n}{m}$ for $a \leq 1$ to obtain Eq.~(\ref{eq:eq5.2aa}) below and then $\binom{n}{m}^2\binom{ln}{lm}^{-1} \leq (\frac{m}{n})^{(l-2)m}$ (see~\cite[Appendix A.5]{shi2012limits}) to obtain Eq.~(\ref{eq:eq5.2bb}), we can write
    \begin{align}
        \displaybreak[0]\sum_{\substack{k = 1 \\ k~ \text{odd}}}^{4m-1}\hat{\|}f\hat{\|}_1^k \sqrt{\sum_{\substack{S\subseteq[n] \\ kd \leq |S| \leq kt}} \frac{\binom{\alpha n/t}{k}^2}{\binom{n}{|S|}} \|\widehat{\rho}(S)\|_{\text{tr}}^2} &\leq \sum_{\substack{k = 1 \\ k~ \text{odd}}}^{4m-1}\hat{\|}f\hat{\|}_1^k \sqrt{\sum_{kd \leq |S| \leq kt} \frac{\binom{\alpha n/t}{k}^2}{\binom{n}{|S|}} \left(\frac{2^{1/(2t)}4tm}{|S|} \right)^{|S|}} \nonumber\\\displaybreak[0]
        &\leq \sum_{\substack{k = 1 \\ k~ \text{odd}}}^{4m-1}\hat{\|}f\hat{\|}_1^k \sqrt{\sum_{kd \leq |S| \leq kt} \left(\frac{\alpha |S|}{kt}\right)^{2k} \frac{\binom{kn/|S|}{k}^2}{\binom{n}{|S|}} \left(\frac{2^{1/(2t)}4tm}{|S|} \right)^{|S|}} \label{eq:eq5.2aa}\\\displaybreak[0]
        &\leq \sum_{\substack{k = 1 \\ k~ \text{odd}}}^{4m-1}\hat{\|}f\hat{\|}_1^k \sqrt{\sum_{kd \leq |S| \leq kt} \left(\frac{\alpha |S|}{kt}\right)^{2k} \left(\frac{|S|}{n}\right)^{|S|-2k} \left(\frac{2^{1/(2t)}4tm}{|S|} \right)^{|S|}} \label{eq:eq5.2bb}\\\displaybreak[0]
        &= \sum_{\substack{k = 1 \\ k~ \text{odd}}}^{4m-1}\hat{\|}f\hat{\|}_1^k \sqrt{\sum_{kd \leq |S| \leq kt} \left(\frac{\alpha n}{kt}\right)^{2k} \left(\frac{2^{1/(2t)}4m}{n/t} \right)^{|S|}}\nonumber\\\displaybreak[0]
        &\leq \sum_{\substack{k = 1 \\ k~ \text{odd}}}^{4m-1}\sqrt{2}\hat{\|}f\hat{\|}_1^k \left(\frac{\alpha n}{kt}\right)^{k} \left(\frac{2^{1/(2t)}4m}{n/t} \right)^{kd/2},\nonumber\displaybreak[0]
    \end{align}
    where we used a geometric series in the last step with $2^{1/(2t)}4tm/n \leq 1/2$ for $n$ sufficiently large and the quantity $\binom{kn/|S|}{k}$ should be interpreted as $\frac{1}{k!}\prod_{j=0}^{k-1}\big(\frac{kn}{|S|}-j\big)$. Finally, and by using that $m \leq \gamma(\epsilon^2/(\alpha \hat{\|}f\hat{\|}_1))^{2/d} (n/t)^{1-2/d}$, we obtain
    \begin{align*}
        \sum_{\substack{k = 1 \\ k~ \text{odd}}}^{4m-1}\hat{\|}f\hat{\|}_1^k \sqrt{\sum_{\substack{S\subseteq[n] \\ kd \leq |S| \leq kt}} \frac{\binom{\alpha n/t}{k}^2}{\binom{n}{|S|}} \|\widehat{\rho}(S)\|_{\text{tr}}^2} &\leq \sum_{\substack{k = 1 \\ k~ \text{odd}}}^{4m-1}\sqrt{2}\hat{\|}f\hat{\|}_1^k \left(\frac{\alpha n}{kt}\right)^{k} \left(\frac{2^{1/(2t)}4\gamma\epsilon^{4/d}(n/t)^{1-2/d}}{\alpha^{2/d}\hat{\|}f\hat{\|}_1^{2/d} n/t} \right)^{kd/2}\\
        &= \sum_{\substack{k = 1 \\ k~ \text{odd}}}^{4m-1}\sqrt{2} \left(\frac{2^{1/(2t)}4\gamma\epsilon^{4/d}}{k^{2/d}} \right)^{kd/2} \leq \frac{\epsilon^2}{2},
    \end{align*}
    taking $\gamma$ sufficiently small in the last step.
    
    \paragraph{Sum II ($4m\leq k \leq \alpha n/t$).} Again using the approximation
	$$
	    \frac{\binom{\alpha n/t}{k}^2}{\binom{n}{kd}} \leq \left(\frac{2\alpha d}{t}\right)^{kd} \frac{\binom{\alpha n/t}{k}^2}{\binom{2\alpha nd/t}{kd}} \leq \left(\frac{2\alpha d}{t}\right)^{kd} \frac{\binom{2\alpha n/t}{k}^2}{\binom{2\alpha nd/t}{kd}}
	$$ 
	($\alpha \leq 1/2$, so $\frac{2\alpha d}{t} \leq 1$), note that the function $h(k) := \big(\frac{2\alpha d}{t}\hat{\|}f\hat{\|}_1^{2/d}\big)^{kd/2} \binom{2\alpha n/t}{k}/\sqrt{\binom{2\alpha nd/t}{kd}}$ is non-increasing in the range $1\leq k \leq \alpha n/t$, since $\frac{2\alpha d}{t}\hat{\|}f\hat{\|}_1^{2/d} \leq 1$ by assumption and then
	\begin{align*}
	    \frac{h(k-1)}{h(k)} \geq \frac{\binom{2\alpha n/t}{k-1}}{\sqrt{\binom{2\alpha nd/t}{kd-d}}}\frac{\sqrt{\binom{2\alpha nd/t}{kd}}}{\binom{2\alpha n/t}{k}} &= \sqrt{\frac{kd}{2\alpha nd/t-kd+d}\prod_{j=1}^{d-1}\frac{2\alpha nd/t-kd+j}{kd-j}}\\
	    &\geq \sqrt{\frac{kd}{2\alpha nd/t-kd+d}\prod_{j=1}^{d-1}\frac{2\alpha nd/t-kd+j+1}{kd-j+1}}\displaybreak[0]\\
		&= \sqrt{\prod_{j=1}^{d-2}\frac{2\alpha nd/t-kd+j+1}{kd-j}} \geq 1,
	\end{align*}
	where we used that $\frac{a}{b} \geq \frac{a+s}{b+s}$ for all $a,b,s>0$ with $a\geq b$ (note that $k \leq \alpha n/t$ implies that $2\alpha nd/t-kd \geq kd$).
	Thus, by using Lemma~\ref{lem:hyper} with $\delta = 1$, the inequality $\binom{n}{m}^2\binom{ln}{lm}^{-1} \leq (\frac{m}{n})^{(l-2)m}$ and that $m \leq \gamma(\epsilon^2/(\alpha \hat{\|}f\hat{\|}_1))^{2/d} (n/t)^{1-2/d}$, we obtain
    \begin{align*}
        \displaybreak[0]\sum_{\substack{k = 4m \\ k~ \text{odd}}}^{\alpha n/t}\hat{\|}f\hat{\|}_1^k \sqrt{\sum_{\substack{S\subseteq[n] \\ kd \leq |S| \leq kt}} \frac{\binom{\alpha n/t}{k}^2}{\binom{n}{|S|}} \|\widehat{\rho}(S)\|_{\text{tr}}^2} &\leq h(4m)\sum_{\substack{k = 4m \\ k~ \text{odd}}}^{\alpha n/t} \sqrt{\sum_{\substack{S\subseteq[n] \\ kd \leq |S| \leq kt}} \|\widehat{\rho}(S)\|_{\text{tr}}^2}\\\displaybreak[0]
        &\leq \frac{\alpha n}{t}2^m h(4m)\\\displaybreak[0]
        &\leq \frac{\alpha n}{t}2^m \left(\frac{2\alpha d}{t}\hat{\|}f\hat{\|}_1^{2/d}\right)^{2md}\left(\frac{2mt}{\alpha n}\right)^{(d-2)2m}\\\displaybreak[0]
        &\leq \frac{\alpha n}{t}2^m \left(\displaybreak[0]\frac{2\alpha d}{t}\hat{\|}f\hat{\|}_1^{2/d}\right)^{2md} \left(\frac{2\gamma \epsilon^{4/d}(n/t)^{1-2/d}}{\alpha^{2/d}\hat{\|}f\hat{\|}_1^{2/d} \alpha n/t}\right)^{(d-2)2m}\\\displaybreak[0]
        &= \frac{\alpha n}{t}2^m \left(\frac{(2d)^d}{t^d}(\alpha\hat{\|}f\hat{\|}_1)^{4/d}\right)^{2m} \left(\frac{2\gamma \epsilon^{4/d}}{(n/t)^{2/d}}\right)^{(d-2)2m}\\\displaybreak[0]
        &\leq \frac{\epsilon^2}{2},\displaybreak[0]
    \end{align*}
    where in the last step we used that $m\geq 1 \implies (1-2/d)4m \geq 1$ (so $n$ is in the denominator and $\epsilon^{(1-2/d)4m} \leq \epsilon$) and picked $\gamma$ sufficiently small.
	
	Summing both results, if $m \leq \gamma(\epsilon^2/(\alpha \hat{\|}f\hat{\|}_1))^{2/d} (n/t)^{1-2/d}$, then $\varepsilon_{bias} \leq \epsilon^2$. 
\end{proof}

\begin{myproof}{Lemma~\ref{lem:lem6.1}}
    We start with the following:
    \begin{align*}
       \sum_{x\in\{-1,1\}^n}\chi_S(x)\mathbf{1}\big[B_f(x,\sigma) = w\big]  &= \sum_{x\in\{-1,1\}^n}\chi_S(x) \prod_{j=1}^{\alpha n/t} \mathbf{1}\left[f(\sigma(x)^{(j)}) = w_j\right]\\
        &= \sum_{x\in\{-1,1\}^n}\chi_{\sigma(S)}(\sigma(x)) \prod_{j=1}^{\alpha n/t} \mathbf{1}\left[f(\sigma(x)^{(j)}) = w_j\right]\\
        &= \sum_{x\in\{-1,1\}^n}\chi_{\sigma(S)}(x) \prod_{j=1}^{\alpha n/t} \mathbf{1}\left[f(x^{(j)}) = w_j\right].
    \end{align*}
    We use the block decomposition $\sigma(S) = [n/t]\bullet \mathbb{U}_{\sigma(S)}$ with $\mathbb{U}_{\sigma(S)} = (U_1,\dots,U_{n/t})\in [t]^{n/t}$. Therefore
    \begin{align*}
        \sum_{x\in\{-1,1\}^n}\chi_S(x)\mathbf{1}\big[B_f(x,\sigma) = w\big] &= \sum_{x\in\{-1,1\}^n}\chi_{[n/t]\bullet \mathbb{U}_{\sigma(S)}}(x) \prod_{j=1}^{\alpha n/t} \mathbf{1}\left[f(x^{(j)}) = w_j\right]\\
        &= \sum_{x\in\{-1,1\}^n}\chi_{U_1}(x^{(1)})\cdots \chi_{U_{n/t}}(x^{(n/t)}) \prod_{j=1}^{\alpha n/t} \mathbf{1}\left[f(x^{(j)}) = w_j\right]\\
        &= 2^{n(1-\alpha)}\sum_{x\in\{-1,1\}^{\alpha n}}\chi_{U_1}(x^{(1)})\cdots \chi_{U_{\alpha n/t}}(x^{(\alpha n/t)}) \prod_{j=1}^{\alpha n/t} \mathbf{1}\left[f(x^{(j)}) = w_j\right],
    \end{align*}
    if $U_j = \emptyset$ for $\alpha n/t < j \leq n/t$, i.e., $\sigma(S) \subseteq[\alpha n]$; otherwise, the sum evaluates to 0.
    %
    \begin{align*}
        \sum_{x\in\{-1,1\}^n}\chi_S(x)\mathbf{1}\big[B_f(x,\sigma) = w\big] &= 2^{n(1-\alpha)} \prod_{j=1}^{\alpha n/t}\left(\sum_{x\in\{-1,1\}^t} \chi_{U_j}(x)\mathbf{1}\left[f(x) = w_j\right]\right)\\
        &= 2^{n(1-\alpha)} \prod_{j=1}^{\alpha n/t}\left(\sum_{x\in\{-1,1\}^t} \chi_{U_j}(x)\frac{1 + f(x)w_j}{2}\right)\\
        &= \frac{2^{n}}{2^{\alpha n/t}}\prod_{j=1}^{\alpha n/t}\left( \mathbf{1}[U_j = \emptyset] + \widehat{f}(U_j) w_j\right).
    \end{align*}
    Since $\operatorname{phdeg}(f) > 0$, $\widehat{f}(\emptyset) = 0$. Hence 
    \begin{align*}
        \sum_{x\in\{-1,1\}^n}\chi_S(x)\mathbf{1}\big[B_f(x,\sigma) = w\big] &= \frac{2^{n}}{2^{\alpha n/t}} \prod_{\substack{j\in[\alpha n/t] \\ U_j \neq \emptyset}}\widehat{f}(U_j) w_j.
    \end{align*}
    By the same token,
    \begin{align*}
         \sum_{x\in\{-1,1\}^n}\chi_S(x)\mathbf{1}\big[B_f(x,\sigma) = \overline{w}\big] &= \frac{2^{n}}{2^{\alpha n/t}} \prod_{\substack{j\in[\alpha n/t] \\ U_j \neq \emptyset}}\widehat{f}(U_j) \overline{w_j}
         = (-1)^{|\mathbb{U}^\ast_{\sigma(S)}|}\frac{2^{n}}{2^{\alpha n/t}} \prod_{\substack{j\in[\alpha n/t] \\ U_j \neq \emptyset}}\widehat{f}(U_j) w_j.
    \end{align*}
    Therefore
    \begin{align*}
        u(\sigma, w, S) = \frac{p_\sigma}{2^{\alpha n/t}} \frac{1 - (-1)^{|\mathbb{U}^\ast_{\sigma(S)}|}}{2} \prod_{\substack{j\in[\alpha n/t] \\ U_j \neq \emptyset}}\widehat{f}(U_j) w_j.
    \end{align*}
    We see that $u(\sigma,w,S)$ can only be nonzero if $\sigma(S)\subseteq[\alpha n]$, if $|\mathbb{U}_{\sigma(S)}^\ast|$ is odd and if $d \leq |U| \leq t$ for all $U \in \mathbb{U}_{\sigma(S)}^\ast$ (since $\operatorname{phdeg}(f)= d$), i.e., if $\sigma(S)\in \Delta$. In summary,
    \begin{align*}
        |u(\sigma, w, S)| = \begin{cases}
	        \frac{p_\sigma}{2^{\alpha n/t}} \prod_{U\in \mathbb{U}^\ast_{\sigma(S)}}|\widehat{f}(U)| &~\text{if}~ \sigma(S)\in \Delta,\\
	        0 &~\text{if}~ \sigma(S)\notin \Delta. \qedalt
        \end{cases}
    \end{align*}
\end{myproof}

\section{Limitations of proof technique}
\label{sec:sec6}

In the last section we saw that Theorem~\ref{thr:thr5.1} guarantees the hardness of the $f\operatorname{-BHP}^{\alpha}_{t,n}$ problem if $f$ has pure high degree $\geq 2$, but the hardness result is not guaranteed if only sign-degree $\geq 2$. To arrive at that result, we used the uniform distribution as a `hard' distribution for Yao's principle. In this section we shall prove that under the uniform distribution we cannot obtain a better result. More specifically, we shall prove that under the uniform distribution there is an efficient bounded-error classical protocol for solving the $f\operatorname{-BHP}^{\alpha}_{t,n}$ problem if $\operatorname{phdeg}(f) \leq 1$.
\begin{theorem}
    \label{thr:nonuniformchap3}
    Under the uniform distribution for Alice and Bob's inputs, if $\operatorname{phdeg}(f) \leq 1$, then $R^1(f\operatorname{-BHP}^{\alpha}_{t,n}) = O(\log{n})$ for a success probability strictly greater than $1/2$ independent of $n$.
\end{theorem}
\begin{proof}
    Assume that $f$ is non-constant, otherwise the result holds trivially. Let $F := \{i\in[t]:\widehat{f}(\{i\}) \neq 0\}$. Given that $\operatorname{phdeg}(f) = 1$, this set is non-empty. Consider the following protocol: Alice picks a subset $I \subseteq [n]$ of indices uniformly at random using shared randomness, where $|I|$ will be determined later, and sends the indices and corresponding bitvalues to Bob. Let $\{x_i\}_{i\in I}$ be the bitvalues sent, and let $j(i) = \lceil \sigma(i)/t\rceil$ and $k(i) \equiv \sigma(i)~(\text{mod}~t)$ for all $i\in I$, where $\sigma\in \mathbb{S}_n$ is Bob's permutation. The probability that none of the indices sent by Alice are matched to a non-zero Fourier coefficient according to Bob's permutation, within one of the $\alpha n/t$ blocks he has, is
    \begin{align*}
        \operatorname{Pr}[k(i)\notin F, ~\forall i\in I] \le \left(1 - \alpha\frac{|F|}{t}\right)^{|I|} \leq e^{-\alpha|I||F|/t},
    \end{align*}
    which we can make almost arbitrarily small by choosing $|I|$ to be sufficiently large. (Note that the first inequality above would be an equality if we chose the elements of $I$ with replacement, and choosing them without replacement cannot make $\operatorname{Pr}[k(i)\notin F, ~\forall i\in I]$ higher.) Hence, with high probability, $\exists i \in I\cap [\alpha n/t]$ such that $k(i)\in F$. Bob computes $\operatorname{sgn}[\widehat{f}(\{k(i)\})]\cdot\sigma(x)^{(j(i))}_{k(i)}\cdot w_{j(i)}$: if it is $+1$, then he outputs that $B_f(x,\sigma) = w$, and if it is $-1$, then he outputs that $B_f(x,\sigma) = \overline{w}$. Otherwise, if $k(i)\notin F$ for all $i\in I\cap [\alpha n/t]$, then Bob outputs a random bit.
   
    To see why the protocol works, we calculate the probability that $\operatorname{sgn}[\widehat{f}(\{k(i)\})]\cdot\sigma(x)^{(j(i))}_{k(i)}$ is equal to $f(\sigma(x)^{(j(i))})$:
    \begin{align*}
        \operatorname*{Pr}_x\left[\operatorname{sgn}[\widehat{f}(\{k(i)\})]\sigma(x)^{(j(i))}_{k(i)} = f(\sigma(x)^{(j(i))})\right] &= \frac{1}{2} + \frac{1}{2^{t+1}}\sum_{x\in\{-1,1\}^t}\operatorname{sgn}[\widehat{f}(\{k(i)\})]\sigma(x)^{(j(i))}_{k(i)} f(\sigma(x)^{(j(i))})\\
        &= \frac{1}{2} + \frac{1}{2} \operatorname{sgn}[\widehat{f}(\{k(i)\})]\cdot\widehat{f}(\{k(i)\})\\
        &= \frac{1}{2} + \frac{1}{2}|\widehat{f}(\{k(i)\})|,
    \end{align*}
    which is greater than $1/2$ and where we used in the first line that the distribution on Alice's inputs is uniform. The overall success probability of the protocol ($\exists i \in I\cap [\alpha n/t]$ such that $k(i)\in F$ and Bob's output equals $f$) is at least $\frac{1}{2} + \frac{1}{2}|\widehat{f}(\{k(i)\})|(1-e^{-\alpha |I||F|/t})$. By taking $|I|=O(1)$, this is strictly greater than $1/2$ by $\Omega\big(\frac{\alpha}{t}|\widehat{f}(\{k(i)\})|\big)$, which does not depend on $n$, since $|\widehat{f}(\{k(i)\})| \geq 2^{1-t}$ (as it is nonzero and is an average of $2^t$ $\pm1$'s).
\end{proof}

\section{Conclusions}

We proposed a very broad generalisation of the famous Boolean Hidden (Hyper)Matching problem, which we called the $f$-Boolean Hidden Partition ($f\operatorname{-BHP}^{\alpha}_{t,n}$) problem. Instead of using the Parity function to arrive at the final bit-string that Alice and Bob wish to explore, we use a generic Boolean function $f$. We partially characterize the communication complexity of the whole problem in terms of one property of $f$: its sign-degree. We proved that if $\operatorname{sdeg}(f) \leq 1$, then there exists an efficient bounded-error classical protocol that solves $f\operatorname{-BHP}^{\alpha}_{t,n}$ with $O(\log{n})$ bits. Similarly to the classical case, we proved that if $\operatorname{sdeg}(f) \leq 2$, then there exists an efficient bounded-error quantum protocol that solves $f\operatorname{-BHP}^{\alpha}_{t,n}$ with $O(\log{n})$ qubits. We then pursued a classical-quantum communication gap by proving classical and quantum lower bounds for cases of the problem where $\operatorname{sdeg}(f) \geq 2$. First we noted that the $f\operatorname{-BHP}^{\alpha}_{t,n}$ problem is hard for almost all symmetric functions with $\operatorname{sdeg}(f) \geq 2$ via a simple reduction from the Boolean Hidden Matching problem. And second we generalised previous communication complexity lower bounds based on Fourier analysis to prove that functions with $\operatorname{phdeg}(f) = d \geq 2$ lead to a classical $\Omega(n^{1-1/d})$ communication cost and functions with $\operatorname{phdeg}(f) = d \geq 3$ lead to a quantum $\Omega(n^{1-2/d})$ communication cost for $f\operatorname{-BHP}^{\alpha}_{t,n}$.

It is known that $\operatorname{phdeg}(f) \leq \operatorname{sdeg}(f)$~\cite{sherstov2011pattern,bun2015dual}, but our lower bounds are probably not tight for {\em all} functions with sign-degree $\geq 2$. We proved that this is an inherent limitation of the chosen distribution for Alice and Bob's inputs during the proof, since under the uniform distribution it is possible to solve the problem with $O(\log{n})$ bits of communication if $\operatorname{phdeg}(f) \leq 1$. We then make the following conjectures.
\begin{conjecture}
    $R^1_\epsilon(f\operatorname{-BHP}^{\alpha}_{t,n}) = \Omega(n^{1-1/d})$ if $\operatorname{sdeg}(f) = d \geq 2$.	
\end{conjecture}
\begin{conjecture}
    $Q^1_\epsilon(f\operatorname{-BHP}^{\alpha}_{t,n}) = \Omega(n^{1-2/d})$ if $\operatorname{sdeg}(f) = d \geq 3$.	
\end{conjecture}
A proof of these results would require a non-uniform distribution on Alice and Bob's inputs, as proven in Theorem~\ref{thr:nonuniformchap3}.

We hope that these conjectures help motivate the development of necessary quantum lower bound techniques.

\subsection*{Acknowledgements}

We would like to thank Ronald de Wolf for pointing out Ref.~\cite{shi2012limits} and for overall suggestions to the paper, Noah Linden for overall suggestions to the paper and Makrand Sinha for useful discussions about the hypercontractive inequality. This project has received funding from the QuantERA ERA-NET Cofund in Quantum Technologies implemented within the European Union's Horizon 2020 Programme (QuantAlgo project); the European Research Council (ERC) under the European Union's Horizon 2020 research and innovation programme (grant agreement No.\ 817581); and EPSRC grants EP/L015730/1, EP/L021005/1, and EP/R043957/1. No new data were created during this study.
	
\bibliographystyle{abbrv}
\bibliography{name}

\end{document}